%% file: sample-sigconf.tex
\documentclass[lettersize,journal]{IEEEtran}

\usepackage{color}
\usepackage{amsmath,amsfonts}
\usepackage{algorithm}
\usepackage{array}
\usepackage{textcomp}
\usepackage{stfloats}
\usepackage{url}
\usepackage{verbatim}
\usepackage{graphicx}
\usepackage{cite}
\usepackage{soul}

    \renewcommand{\baselinestretch}{0.945}
\newcommand{\stitle}[1]{\vspace{1ex} \noindent{\bf #1}}
\usepackage{subfigure}
\usepackage{multirow}
\usepackage{multicol}
\usepackage{booktabs} 
\usepackage{amssymb}
\usepackage{amsmath}

\newtheorem{theorem}{Theorem}
\newtheorem{definition}{Definition}

\newtheorem{example}{Example}

\newenvironment{proof}{\quad{\it Proof:}}{\hfill $\square$\par}  
\usepackage{subfigure}
\usepackage[shortlabels]{enumitem} 
\usepackage{multirow}
\usepackage[noend]{algpseudocode}
\usepackage{algorithmicx,algorithm}

\begin{document}

\pagestyle{empty}

\title{Effective and Efficient Conductance-based Community Search at Billion Scale} 	

\author{Longlong Lin, Yue He, Wei Chen, Pingpeng Yuan, Rong-Hua Li, 
 Tao Jia
\IEEEcompsocitemizethanks{\IEEEcompsocthanksitem Longlong Lin, Yue He, Wei Chen, Tao Jia are with the College of
		Computer and Information Science,
		Southwest University, China. Email: longlonglin@swu.edu.cn; yue\_he\_1222@163.com; WeiChen605@126.com; tjia@swu.edu.cn.  Tao Jia is the corresponding author.
		\IEEEcompsocthanksitem  Pingpeng Yuan is with the School of Computer Science and Technology, Huazhong University of Science and Technology, China.
		E-mail: ppyuan@hust.edu.cn
		\IEEEcompsocthanksitem  Rong-Hua Li is with the School of
		Computer Science and Technololgy,
		Beijing Institute of Technology, China. Email: lironghuabit@126.com \protect\\}
	\thanks{Manuscript received XXX, XXX; revised XXX, XXXX.}}

\markboth{Journal of \LaTeX\ Class Files,~Vol.~14, No.~8, August~2021}%
{Shell \MakeLowercase{\textit{et al.}}: A Sample Article Using IEEEtran.cls for IEEE Journals}


\maketitle

\begin{abstract}
Community search is a widely studied semi-supervised graph clustering problem, retrieving a high-quality connected subgraph containing the user-specified query vertex.  However, existing methods primarily focus on cohesiveness within the community but ignore the sparsity outside the community, obtaining sub-par results. Inspired by this, we adopt the well-known \emph{conductance} metric to measure the quality of a community and introduce a novel problem of \emph{conductance}-based community search (\emph{CCS}). \emph{CCS} aims at finding a subgraph with the smallest \emph{conductance} among all connected subgraphs that contain the query vertex. We prove that the \emph{CCS} problem is NP-hard. To efficiently query \emph{CCS}, a four-stage \emph{subgraph-conductance}-based community search algorithm, \emph{SCCS}, is proposed. Specifically, we first greatly reduce the entire graph using local sampling techniques. Then, a three-stage local optimization strategy is employed to continuously refine the community quality. Namely, we first utilize a seeding strategy to obtain an initial community to enhance its internal cohesiveness. Then, we iteratively add qualified vertices in the expansion stage to guarantee the internal cohesiveness and external sparsity of the community. Finally, we gradually remove unqualified vertices during the verification stage. Extensive experiments on real-world datasets containing one billion-scale graph and synthetic datasets show the effectiveness, efficiency, and scalability of our solutions.

\end{abstract}

\input{introduction}

\input{problem_formulation}

\input{algorithm}
\input{experiment}
\input{related_work}

\section{CONCLUSION}\label{sec:conclude}

In this work, we introduce a new problem called \textit{conductance}-based community search (\emph{CCS}), aiming to find the subgraph with the minimum \textit{conductance} among all connected subgraphs containing a given query vertex. We establish the NP-hardness of \emph{CCS}. To address this, we first propose a basic feasible algorithm \emph{PPRCS} with several disadvantages. Then, we propose an advanced four-stage \textit{subgraph-conductance}-based algorithm \emph{SCCS}. Finally, extensive experiments on six real-world datasets  containing one billion-scale graph and synthetic graphs show that our solutions are superior and competitive with the existing six competitors. 

While our algorithms demonstrate strong performance in both time efficiency and solution quality, they remain heuristic in nature. Given that \textit{CCS} is proven to be NP-hard, future work will focus on exploring approximation ratios and worst-case performance bounds. Furthermore, we intend to implement a parallel approach to enhance the scalability of our algorithms for larger networks, thereby facilitating more effective applications in real-world scenarios. 

\section{ACKNOWLEDGMENTS}
The work was supported by (1) the Natural Science Foundation of China (62402399, 72374173), (2) Fundamental Research Funds for the Central Universities (No. SWU-XDJH202303),
 (3) Chongqing Innovative Research Groups (No. CXQT21005),
(4) the High Performance Computing clusters at Southwest University.

\bibliographystyle{IEEEtran}
\bibliography{sample}
\vspace{-10mm}
\begin{IEEEbiography}[{\includegraphics[width=1in,height=1.25in,clip,keepaspectratio]{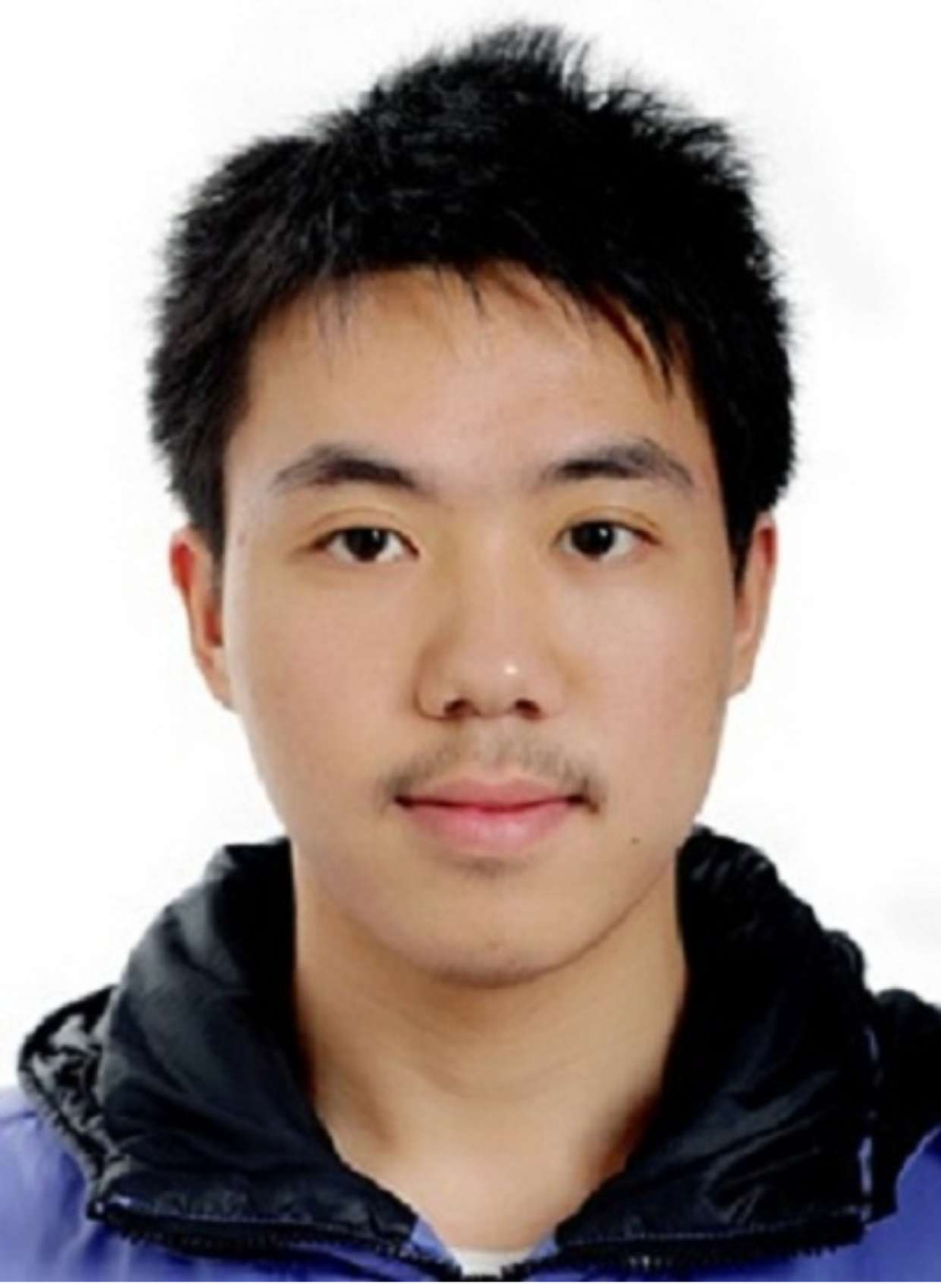}}]{Longlong Lin} received his Ph.D. degree from Huazhong University of Science and Technology (HUST), Wuhan, in 2022. He is currently an associate professor in the College of
	Computer and Information Science, Southwest University, Chongqing. His current research interests include graph clustering and graph-based machine learning.
\end{IEEEbiography}
\vspace{-10mm}

\begin{IEEEbiography}[{\includegraphics[width=1in,height=1.25in,clip,keepaspectratio]{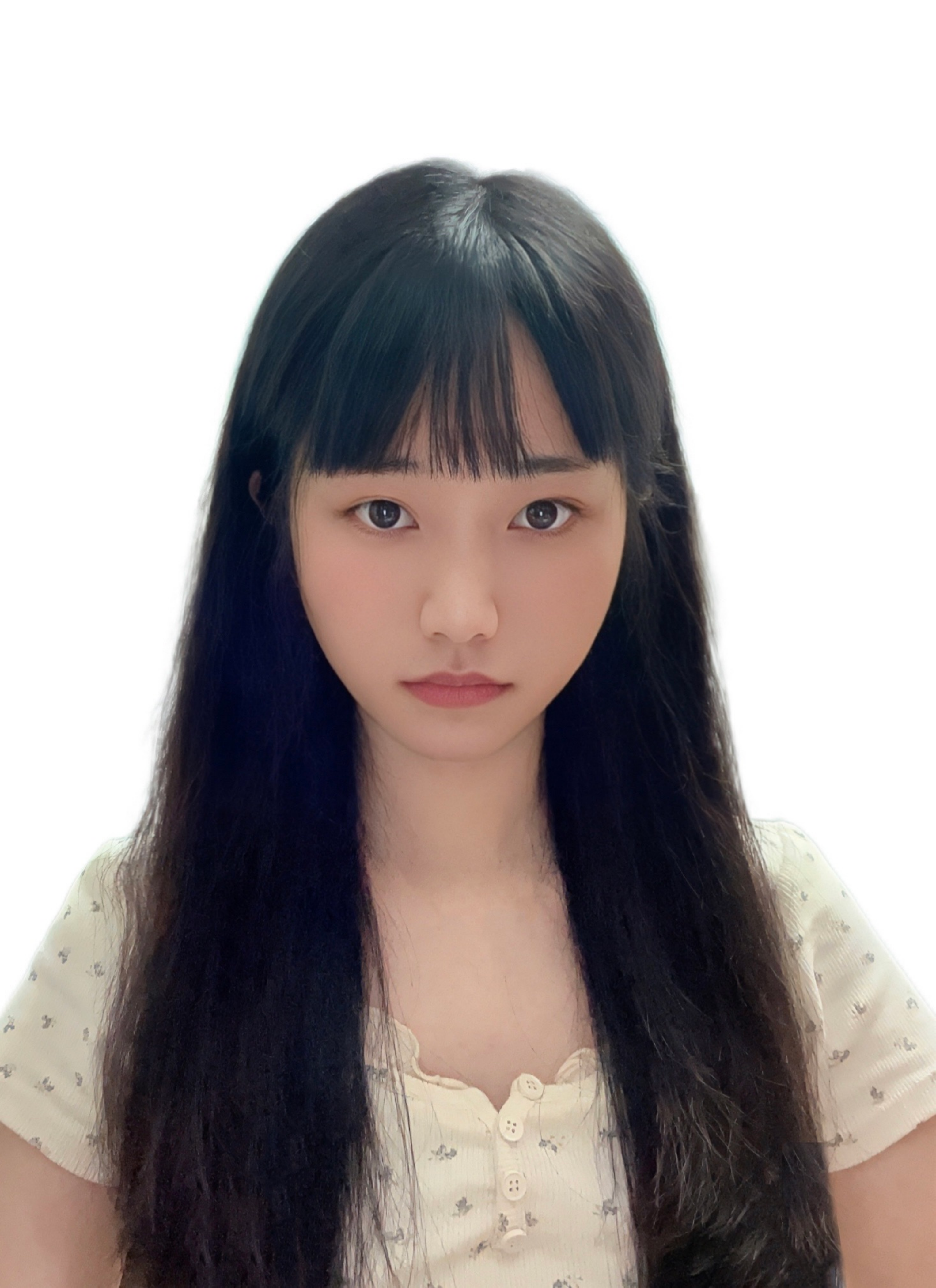}}]{Yue He} received her B.E. degree from Chengdu Normal University, China, in 2022. She is currently pursuing an M.S. degree at the College of Computer and Information Science,  Southwest University, China.  Her current research interests include community search and data mining.
\end{IEEEbiography}
\vspace{-10mm}

\begin{IEEEbiography}[{\includegraphics[width=1in,height=1.25in,clip,keepaspectratio]{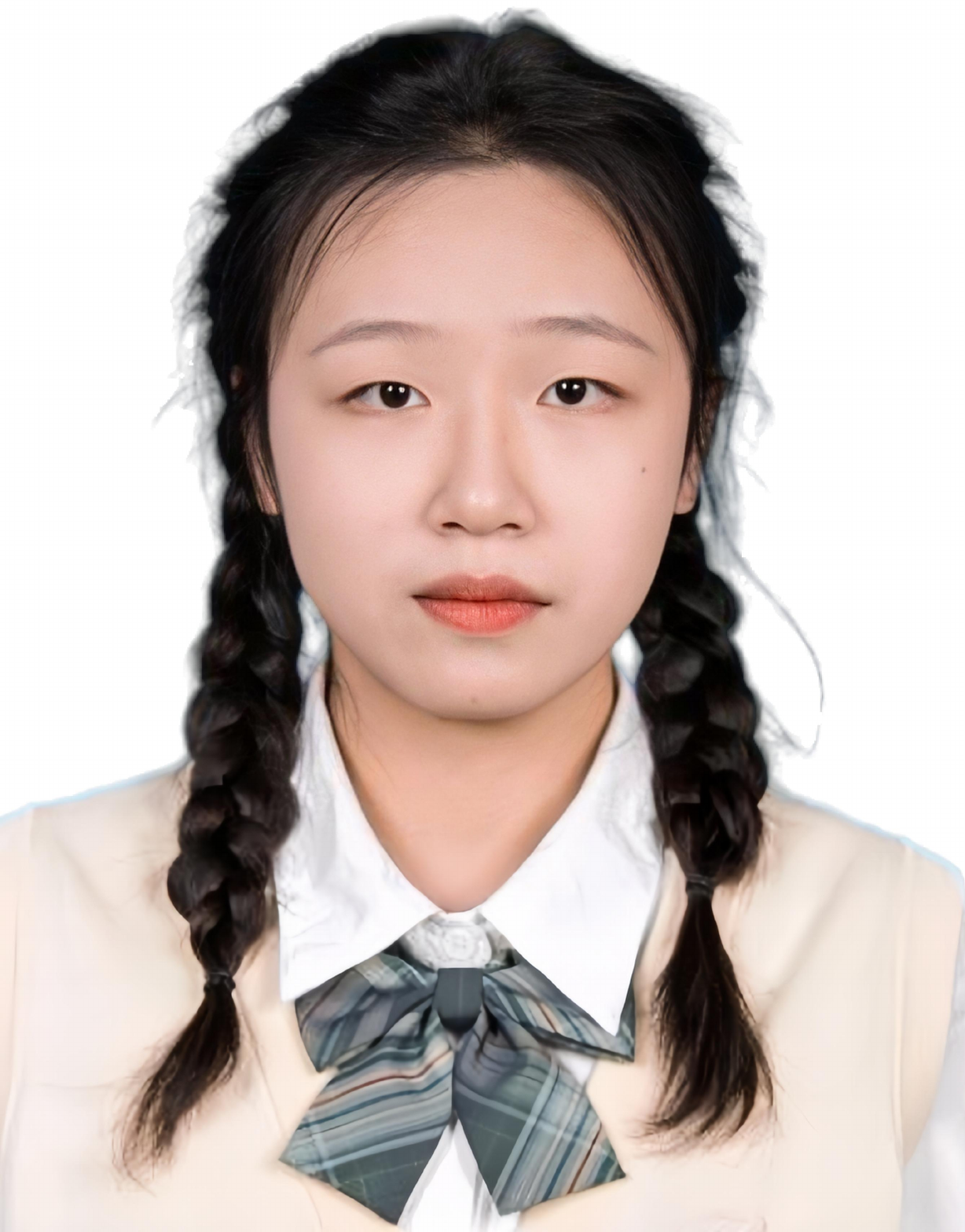}}]{Wei Chen} received her B.E. degree from Sichuan Agricultural University, China, in 2024. She is currently pursuing an M.S. degree at the College of
Computer and Information Science, Southwest University. Her current research interests include dense subgraph mining and community search.
\end{IEEEbiography}
\vspace{-10mm}

\begin{IEEEbiography}[{\includegraphics[width=1in,height=1.25in,clip,keepaspectratio]{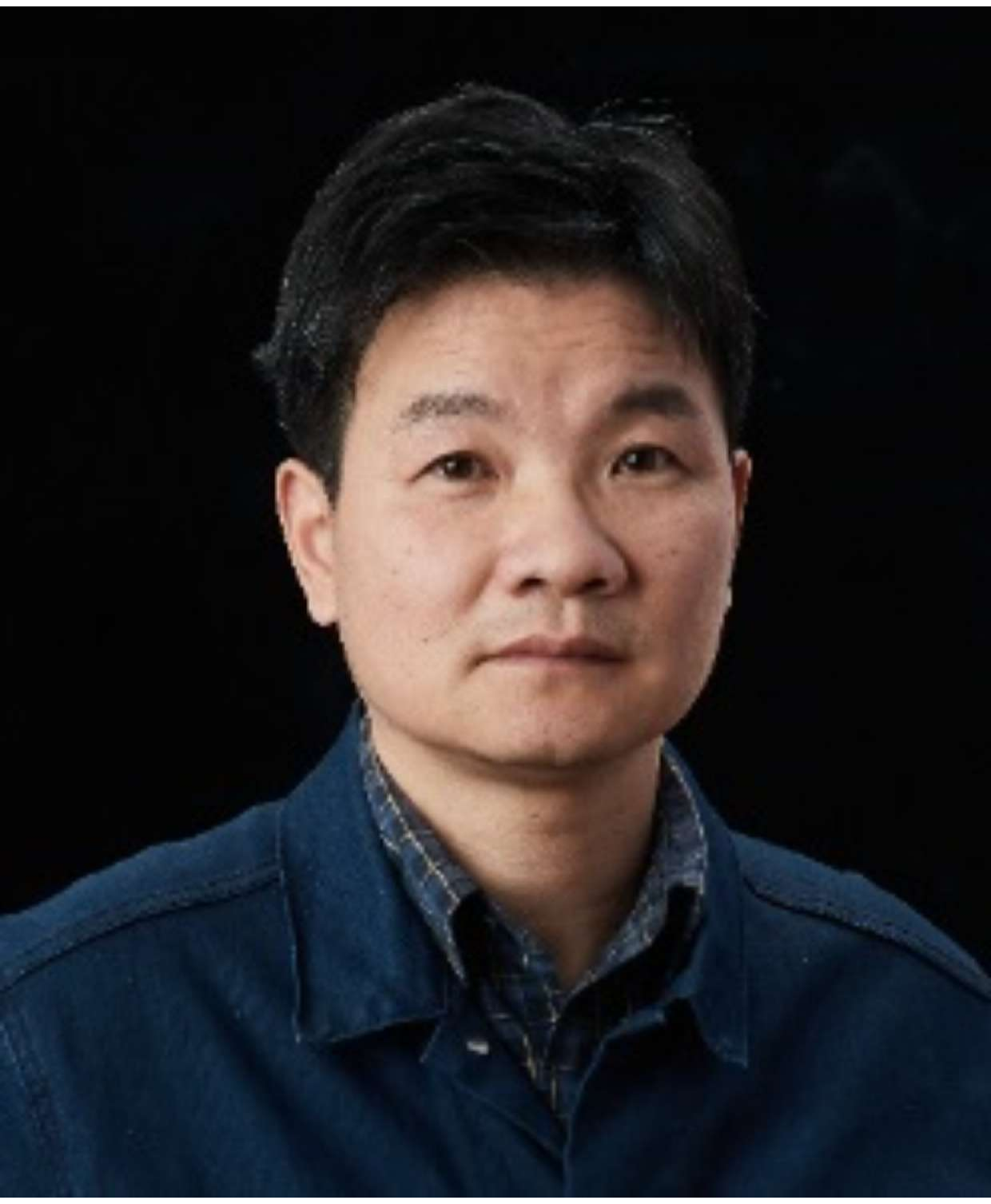}}]{Pingpeng Yuan}
	received his Ph.D. degree in computer science from Zhejiang University, Hangzhou, in 2002. He is now a professor in the School of Computer Science and Technology at Huazhong University of Science
	and Technology (HUST), Wuhan. His research interests include databases, knowledge representation and reasoning, and natural language processing, with a focus on high performance computing. He is the principle developer in multiple system prototypes, including TripleBit and PathGraph.
\end{IEEEbiography}
\vspace{-10mm}

\begin{IEEEbiography}[{\includegraphics[width=1in,height=1.25in,clip,keepaspectratio]{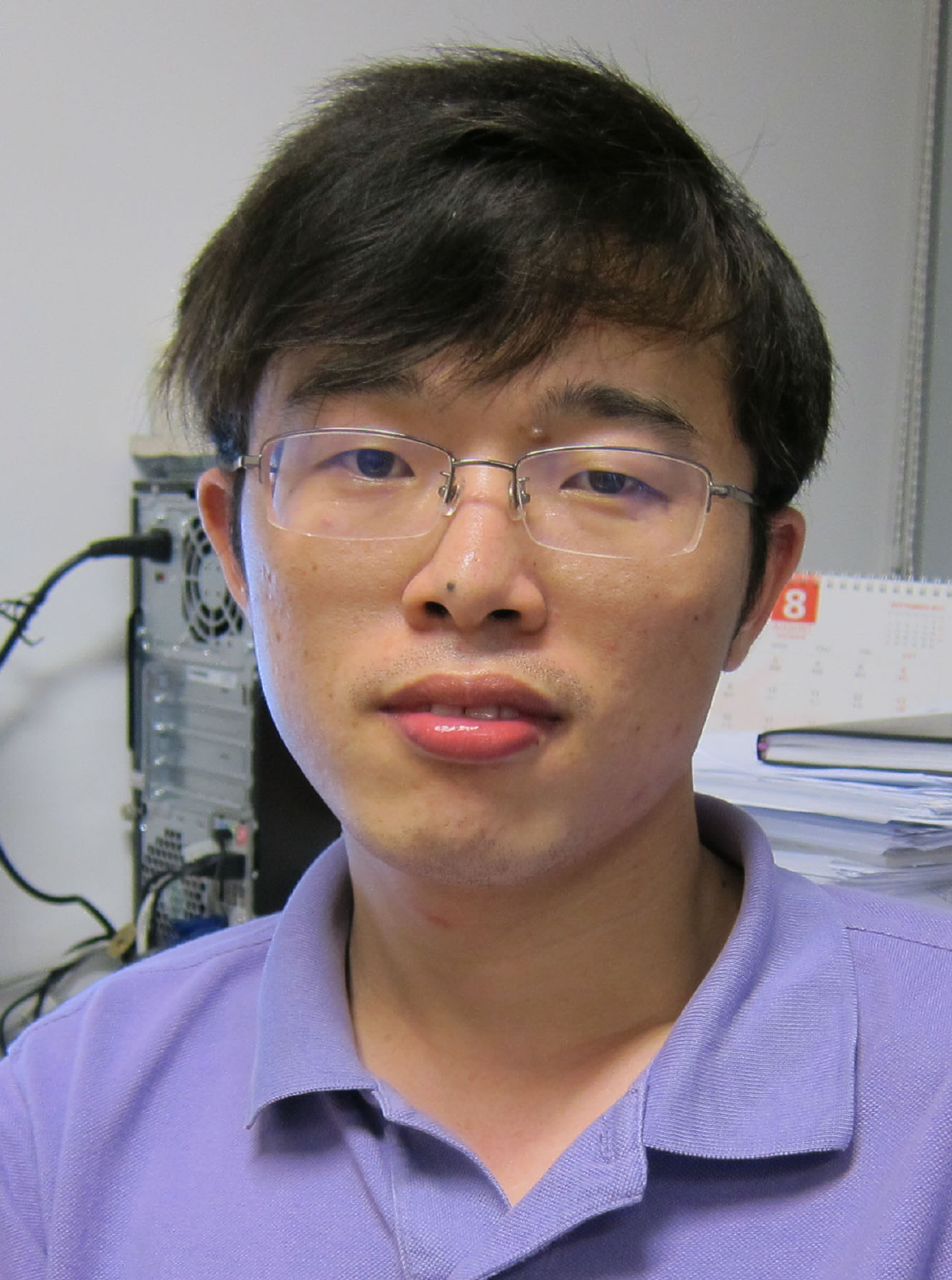}}]{Rong-Hua Li}
	received the PhD degree from the Chinese University of Hong Kong, in 2013. He is currently a professor with the Beijing Institute of Technology (BIT), Beijing, China. Before joining BIT in 2018, he was an assistant professor with Shenzhen University. His research interests include graph data management and mining, social network analysis, graph computation systems, and graph-based machine learning.
\end{IEEEbiography}
\vspace{-10mm}

\begin{IEEEbiography}[{\includegraphics[width=1in,height=1.25in,clip,keepaspectratio]{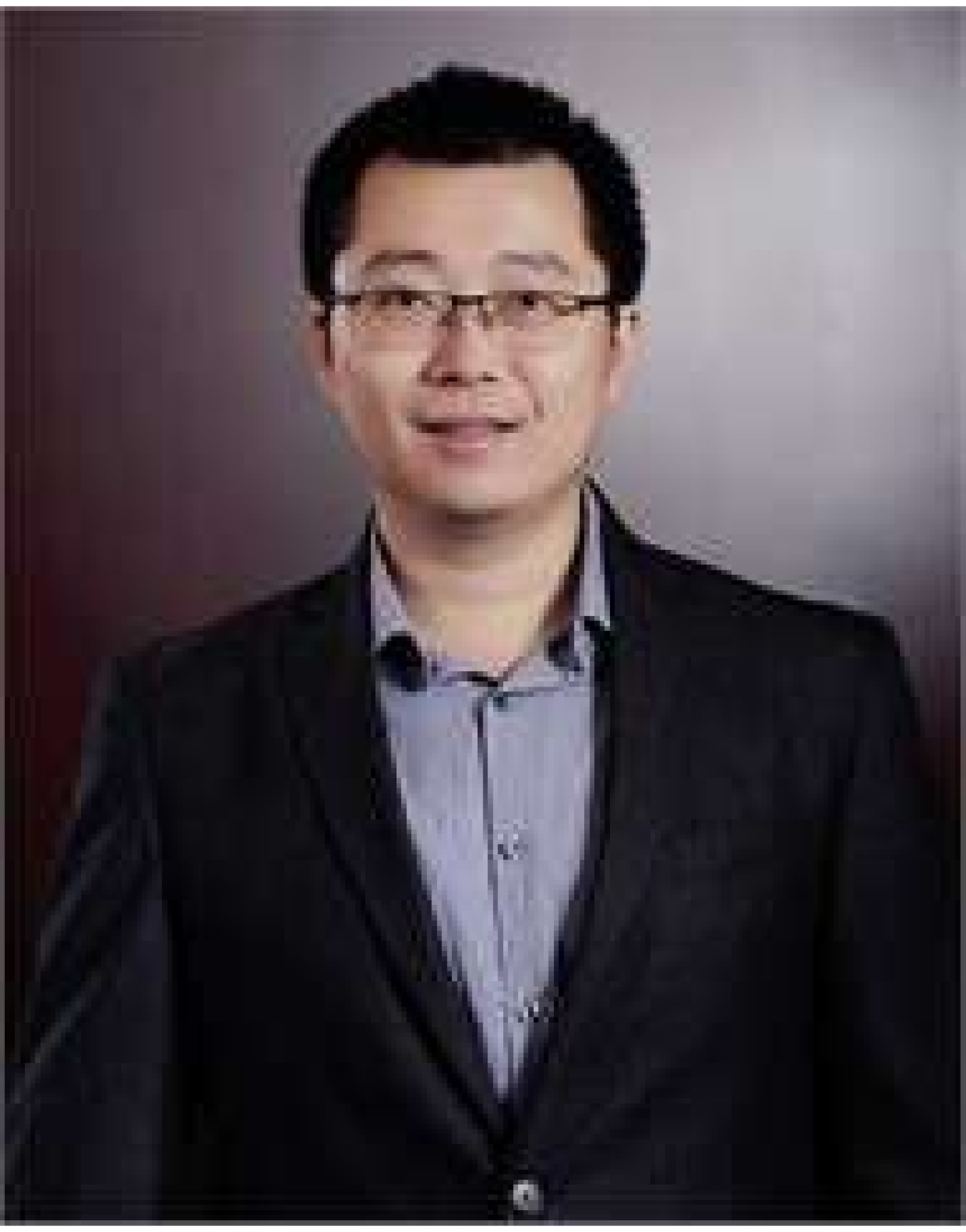}}]{Tao Jia} received the BSc degree from Nanjing University, China. He received his MSc and PhD degree from Virginia Tech, USA. He is currently a Professor at Southwest University, China. His research interest includes graph mining, brain networks, and social computing.
\end{IEEEbiography}
\vspace{-10mm}

\end{document}

%% file: introduction.tex
\section{Introduction} \label{sec:introduction}
Community search, identifying a high-quality connected subgraph (known as community) containing the given query vertex, is one of the fundamental tasks in network science \cite{DBLP:conf/focs/AndersenCL06,DBLP:conf/kdd/SozioG10, DBLP:conf/sigmod/CuiXWW14,   DBLP:conf/sigmod/YangXWBZL19, DBLP:conf/kdd/ChenL0XY020}.  The problem also witnesses many practical applications, including social recommendation \cite{DBLP:conf/kdd/SozioG10}, protein complexes identification \cite{DBLP:conf/sigmod/CuiXWW14}, and  impromptu activities organization \cite{DBLP:conf/kdd/ChenL0XY020}.  Therefore, many community search models have been proposed in the literature, notable examples include $k$-core based \cite{DBLP:conf/kdd/SozioG10,DBLP:conf/sigmod/CuiXWW14, DBLP:journals/datamine/BarbieriBGG15} and $k$-truss based  \cite{DBLP:conf/sigmod/HuangCQTY14, DBLP:conf/sigmod/LiuZ0XG20}. 

\begin{figure}[t!]
\vspace*{0.6cm}
\centering
\includegraphics [width=0.4\textwidth] {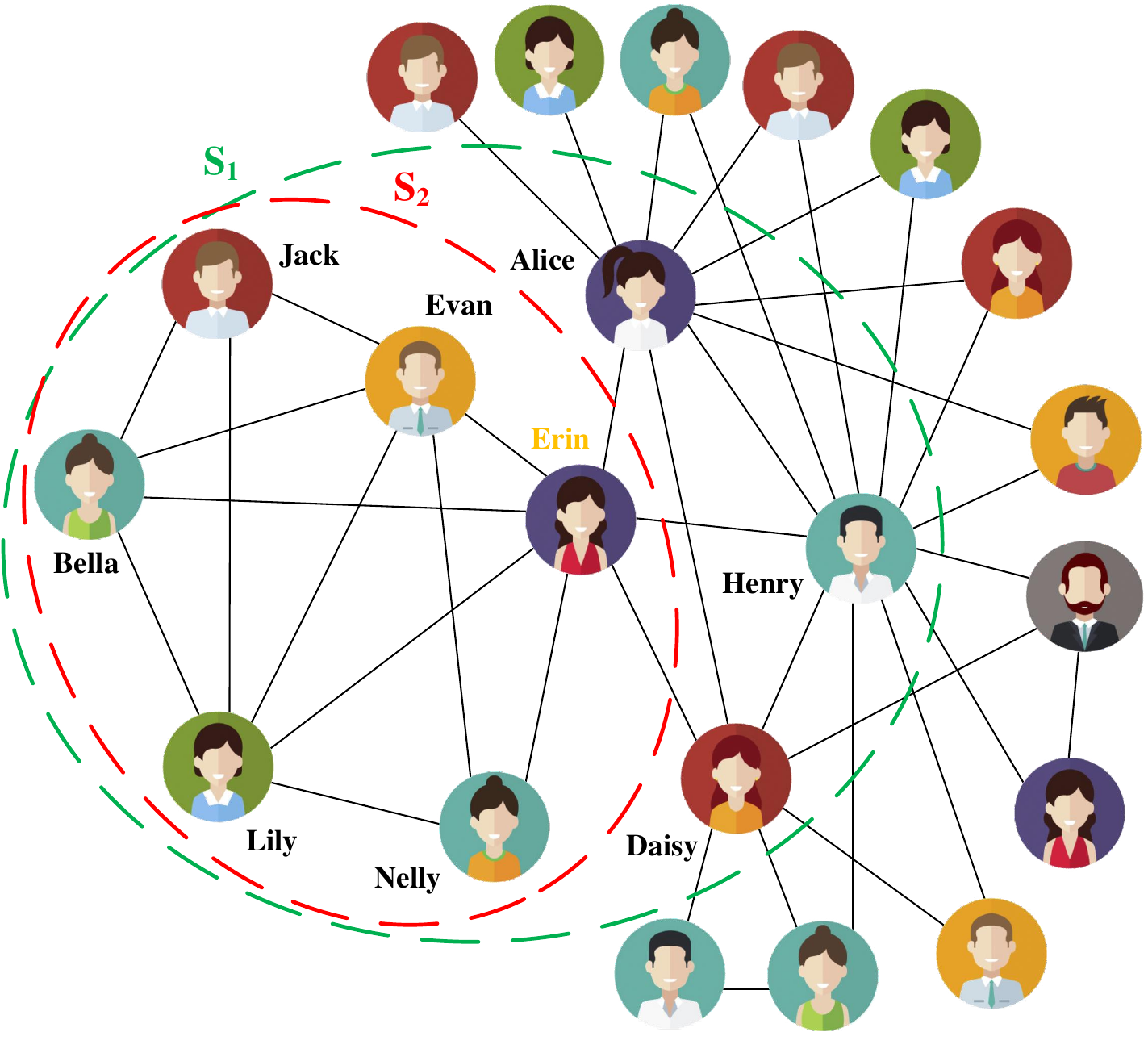}
\caption{Motivation example. $q=Erin$ is the query vertex. $S_{1}$ is the answer returned based on $k$-core \cite{DBLP:conf/sigmod/CuiXWW14} and $k$-truss  \cite{DBLP:conf/sigmod/HuangCQTY14}. $S_{2}$ is our answer.} 
\label{fig:motivation} \vspace*{-0.6cm}
\end{figure}

Despite their success, most existing state-of-the-art (SOTA) methods are based on  cohesive subgraph models, which \emph{only} focus on the internal cohesiveness  of the output community \cite{DBLP:journals/vldb/FangHQZZCL20}. Taking Figure \ref{fig:motivation} as an example, the query vertex $Erin$ plans to organize a seminar, $S_{1} = \{Erin,Evan,Jack,Bella,Lily,Nelly,Alice,Henry,Daisy\}$ is the $3$-core community containing $Erin$ (i.e., each vertex has at least 3 neighbors in the community. Note that
4-core containing $Erin$ does not exist in Figure \ref{fig:motivation}). Similarly, $S_{1}$ is also the $4$-truss community containing $Erin$  (i.e., each edge has at least 4-2 triangles in the community. Note that
5-truss containing $Erin$ does not exist in Figure \ref{fig:motivation}). However, according to \cite{DBLP:conf/www/LeskovecLM10, DBLP:journals/kais/YangL15}, we can know that a community is high-quality if the community is densely connected internally and well separated from the remainder of the graph. Thus, cohesive subgraph-based community search models significantly overlook the unique external separability of communities, potentially obtaining poor results, as stated in our experiments. For example, $S_{1}$ has greater external connections (20 edges) than internal connections (18 edges), indicating that $S_{1}$ is not ideally separated from the rest of the network \cite{DBLP:conf/www/LeskovecLM10, DBLP:journals/kais/YangL15}. On top of that, the qualities of communities returned by $k$-core or $k$-truss are seriously compromised by the input parameter $k$. For example, to obtain smaller communities that better align the size (usually between 3 to 100 in Table \ref{tab:range}) of the ground truths,  they tend to maximize $k$. However, since the low core/truss numbers of most vertices in real-world networks \cite{DBLP:journals/kais/ShinEF18}, the maximum $k$ is still naturally small. As a result, the resultant community (returned by $k$-core or $k$-truss) sizes remain significantly larger than those of real communities (Table \ref{tab:size}), failing to meet the properties of real-world networks \cite{DBLP:conf/www/LeskovecLDM08,DBLP:journals/tbd/WangYBH24}. Last but not least, existing solutions incur substantial computational costs for handling massive graphs with billions of edges. For example, on a Friendster with 1.8 billion edges, \textit{CSM}, \textit{TCP} and \textit{HK-Relax} cannot obtain results within five hours in our experiments (Table \ref{tab:running}), which is obviously impractical for graph exploration. Summing up, it is meaningful and challenging to improve the efficiency and quality of existing community search methods.


Intuitively, $S_{2}=\{Erin,Evan,Jack,Bella,Lily,Nelly\}$ may be an ideal community in
Figure \ref{fig:motivation}. This is because $S_2$ has 12 dense internal edges and only has 3 sparse external edges. Inspired by this, 
in this paper, we use the representative and \emph{free-parameter} metric \emph{conductance} \cite{DBLP:conf/focs/AndersenCL06, DBLP:conf/icdm/BianYCWLZ18, DBLP:conf/sigmod/YangXWBZL19} to measure the quality of a community. Given an undirected graph $G$ and a vertex set $C$, the \emph{conductance} of $C$ is defined as $\phi(C)=\frac{|E(C,\bar{C})|}{\min\{vol(C), 2m-vol(C)\}}$, in which $|E(C,\bar{C})|$ is the number of edges with one endpoint in $C$ and another not in $C$, $vol(C)$ is the sum of the degree of all vertices in $C$, and $m$ is the number of the edges in $G$. Thus, a smaller $\phi(C)$ implies that  the number of edges going out of $C$ is relatively small compared with the number of edges within $C$. As a consequence, the smaller the $\phi(C)$, the better the quality of the community $C$ \cite{DBLP:conf/www/LeskovecLM10, DBLP:journals/kais/YangL15}, which reflects the cohesion inside the community and the sparseness outside.  For example, $\phi(S_{2})=0.111 < \phi(S_{1})= 0.833$ indicates that $S_{2}$ is more suitable for $Erin$ to organize a seminar.

As stated in Sections \ref{sec:gr} and \ref{sec:relate}, the diffusion-based local clustering is most relevant to ours, which aims to identify a small \emph{conductance}  community near the given query vertex \cite{DBLP:conf/aaai/LinLJ23}. Although the local clustering can obtain theoretically good graph cuts, it is not desirable in community search. For example, the community returned by the local clustering may not be connected \cite{DBLP:conf/icml/ZhuLM13} or contain no the given query vertex \cite{DBLP:conf/webi/LuoWP06}, which contradicts the definition of community search \cite{DBLP:journals/vldb/FangHQZZCL20}. To address the dilemma, we introduce a novel \emph{conductance}-based community search (named as \emph{CCS}), which aims at finding a subgraph with the smallest \emph{conductance} among all connected subgraphs containing the given query vertex.

\stitle{Our solutions.} Addressing the \emph{CCS} problem raises significant challenges due to its NP-hardness (Theorem \ref{thm:nph}), thus we resort to the heuristic  solutions. One possible simple solution is to adjust the existing local clustering framework and propose an efficient heuristic algorithm \emph{PPRCS} to solve our \emph{CCS}.  Specifically, \emph{PPRCS} first is to compute the Personalized PageRank vector seeding by the query vertex $q$. Then, \emph{PPRCS} executes a modified sweep cut over the vector, ensuring the resultant community is connected and contains $q$. However, the community quality of such naive solution \emph{PPRCS}  is heavily dependent on many hard-to-tune parameters, resulting in that its performance is unstable and in most cases very poor \cite{DBLP:conf/icml/ZhuLM13,DBLP:conf/kdd/KlosterG14}. To overcome these limitations, we propose an efficient and effective four-stage \emph{subgraph-condutance}-based algorithm \emph{SCCS}. Specifically, in the first stage, \emph{SCCS} samples a subgraph (the subgraph is much smaller than the initial graph) as the subsequent search space, significantly saving computational costs. In the next three stages, we adopt the expansion and contraction paradigm. Concretely, in the initial community stage (the second stage), we consider the largest clique containing $q$ as the initial community to ensure that the internal connections are sufficiently tight. In the third stage, we iteratively add qualified vertices that increase the community quality (Definition \ref{def:f(S)}) until the termination condition is met, ensuring the community remains tightly connected internally while being sparsely connected externally. Finally, in the verification stage (the fourth stage), we use boundary-based pruning techniques to iteratively remove unqualified vertices, further improving the quality of the community. Note that the last two stages execute interactively. In a nutshell, our main contributions are highlighted as follows:

\begin{itemize} [leftmargin=8pt, topsep=0pt]
\item \textbf{\underline{Novel Model.}} We formulate the problem of \emph{conductance}-based community search to simultaneously consider both cohesion within the community and separation from the outside. This is an improvement over existing methods that mainly focus on internal cohesion. We also prove that the \emph{CCS} problem is NP-hard.

\item \textbf{\underline{Efficient Algorithms.}} We propose a basic algorithm \emph{PPRCS} and an advanced four-stage \emph{subgraph-conductance}-based algorithm \emph{SCCS}. One striking feature of \emph{SCCS} lies in its ability to restrict the search space to a very small subgraph by sampling, thus mitigating prohibitively high time overhead, and the clustering results are more in line with the properties of real-world networks.
	
\item \textbf{\underline{Extensive Experiments.}} Experimental evaluation on six real-world datasets with ground-truth communities and synthetic LFR benchmark datasets shows that our proposed solutions are indeed more efficient, effective, and scalable compared with six competitors. For instance, on the Friendster dataset with 1.8 billion edges, our  \emph{PPRCS} (resp., \emph{SCCS}) takes only about 0.6 (resp., 11) seconds, while most other methods fail to produce results within five hours. Besides,  our solutions have better the \textit{F1-score} when compared with the competitors.  For reproducibility, we release our source codes and datasets at https://github.com/yuehe1222/SCCS.
\end{itemize}

%% file: problem_formulation.tex
\section{Problem Formulation} \label{sec:pro}

\subsection{Notations}
We use $G(V, E)$ to denote an undirected graph, in which $V$ (resp. $E$) indicates the vertex set (resp. the edge set) of $G$. Let $n=|V|$ and $m=|E|$ be the number of vertices and edges in $G$, respectively. For a vertex subset $S\subseteq V$, we define $G_{S}=(S,E_S)$ as the subgraph induced by $S$ if $E_S= \{(u,v)\in E| u,v \in S\}$, and $\bar{S}=V\setminus S$ as the complement of $S$. We use $E(S, \bar{S})=\{(u,v) \in E|u\in S, v\in \bar{S}\}$ to represent the edges with one endpoint in $S$ and another not in $S$. Let $N_{S}(u)=\{v\in S|(u,v)\in E\}$ be the neighbors of the vertex $u$ in $S$ and $d_{S}(u)=|N_{S}(u)|$ indicates the degree of $u$ in $S$.  When the context is clear, we abuse $N(u)$, $d(u)$, and $S$ to represent $N_{V}(u)$, $d_{V}(u)$, and $G_S$, respectively. Let $vol(S)=\sum_{u \in S}d(u)$ be the sum of the degree of all vertices in $S$. Table \ref{tab:notaion} summarizes the  frequently-used symbols.

A cluster is a vertex set $S \subseteq V$. According to \cite{DBLP:conf/www/LeskovecLM10,DBLP:journals/kais/YangL15}, we know that a cluster $S$ is good if the cluster is densely connected internally and well separated from the remainder of $G$. Therefore, we use a representative and free-parameter metric \emph{conductance} \cite{DBLP:conf/icdm/BianYCWLZ18, DBLP:conf/sigmod/YangXWBZL19, DBLP:conf/focs/AndersenCL06, DBLP:conf/kdd/KlosterG14,DBLP:conf/kdd/Lin0WZL24} to measure the quality of a cluster $S$.

\begin{definition} [Conductance $\phi$] \label{def:co}
	Given an undirected graph $G(V,E)$ and a vertex subset $S \subseteq V$, the conductance of $S$ is defined as	$\phi(S)=\frac{|E(S,\bar{S})|}{\min\{vol(S), 2m-vol(S)\}}$. 
\end{definition}

By Definition \ref{def:co}, we have $\phi(S)=\phi(\bar{S})$. Figure \ref{fig:problem_statement} illustrates the concept of \textit{conductance} on a simple synthetic graph.  Next, we formally state our problem as follows.

\setcounter{table}{0}
\begin{table}[t!] 
\centering
\caption{The frequently-used symbols} 
\begin{tabular}{c|c}
\hline
Symbol &Definition\\ 
\hline
$G$ & an undirected graph\\
\hline
$S_G$ &a sampling subgraph induced by $G$\\
\hline
$n,m$ & the number of vertices and edges in $G$\\
\hline		
$E(S, \bar{S})$  & an edge set with only one endpoint in $S$\\ 
\hline
$vol(S)$  & the sum of the degree of all vertices in $S$\\ 
\hline
$N_S(u), d_S(u)$ & the neighbors and degree of $u$ in $S$\\
\hline
$d_{G}(u, q)$ & the depth from $q$ to $u$\\ 
\hline
$\alpha ,r_{max}$ & the jump probability and termination threshold\\ 
\hline
$dp$ &the sampling depth\\
\hline
$h, l$ & the upper and lower bounds of sampled vertices\\  
\hline
$count$ &  the upper limit of one expansion\\
\hline
$A_i$ & the adjacency set of size $i$\\
\hline
$B_j$ & the boundary set of size $j$\\
\hline
$\Delta f(S)$ & the quality gain of $S$\\
\hline
\end{tabular}
\label{tab:notaion}
\end{table}

\stitle{Problem Statement (\emph{CCS}).} Given an undirected  graph $G(V,E)$ and a query vertex $q \in V$, the \underline{C}onductance-based \underline{C}ommunity \underline{S}earch (\emph{CCS}) problem aims to identify a subgraph $S$ satisfying,

\textbf{(1) Seed Inclusion:} $q\in S$; 

\textbf{(2) Connected:} $S$ is a connected subgraph; 

\textbf{(3) High quality:} $\phi(S)$ is minimized among all possible choices that satisfy the above two conditions. Namely, for any $S'\subseteq V$ and $q \in S'$, we have $\phi(S)\leq \phi(S')$.

The \textit{CCS} problem aims to find a connected subgraph $S$ containing the query vertex $q$ that is internally condensed and externally sparse. We use \textit{conductance} to measure cohesion and exosparcness of $S$, that is, the smaller the \textit{conductance} of $S$, the higher the quality.

\begin{example}
Consider Figure \ref{fig:problem_statement} and the query vertex is $v_0$ or $v_7$. According to the problem definition of our \emph{CCS}, $S_1=\{v_0, v_1, v_2, v_3\}$ denotes the optimal community containing $q=v_0$ with the lowest \emph{conductance} of $\phi^*=1/11$. Similarly, $S_2=\{v_7, v_8, v_9,v_{10}\}$ represents the optimal cluster containing $q=v_7$ with the lowest \emph{conductance} of $\phi^*=1/13$.
\end{example}

\stitle{Remark.} Our proposed \emph{CCS} and \cite{DBLP:journals/eswa/HeLYLJW24} both use conductance as the objective function to identify the community to which the query node belongs; however, there are several important differences between them, which are summarized as follows. (1) \cite{DBLP:journals/eswa/HeLYLJW24} is the first to use \textit{conductance} to investigate the problem of community search with size constraints, introducing a new model for \textit{conductance}-based community search with size constraints (named as \textit{CCSS}). In contrast, our manuscript explores the problem of \textit{conductance}-based community search without size constraints (named as \textit{CCS}), making it more flexible. For example, \textit{CCSS} has two size parameters, which
makes choosing the proper parameters a challenge for users. In particular, if the parameter settings are too strict, it may obtain empty results. On the other hand, if the parameters are too loose, the quality of the results may not be excellent. We emphasize that community search with size constraints and community search without size constraints represent two orthogonal directions that have been extensively studied by numerous literature (for more details, see the related work in [1]). (2) The algorithms designed by \textit{CCSS} are based on a two-level loop, which includes two key technologies: the vertex scoring function and the perturbation strategy. The vertex scoring function is to obtain coarse-grained candidate communities, while the perturbation strategy aims to refine the quality of previous candidate communities. In our manuscript, we propose a basic algorithm based on Personalized PageRank and an advanced four-stage \textit{subgraph-conductance}-based algorithm \textit{SCCS}.  The standout feature of \textit{SCCS} lies in its ability to restrict the search space to a very small subgraph through sampling, followed by optimizing the community using alternating expansion and contraction strategies.

\subsection{Problem Hardness}

\begin{theorem}\label{thm:nph}
Given an undirected graph $G(V, E)$ and a query vertex $q \in V$, the \emph{CCS} problem is NP-hard.
\end{theorem}

\begin{proof}
Given an undirected graph $G$  and an integer $k$, if there exists a subgraph of size $k$ in $G$ such that any two vertices are connected to each other. This is called the $k$-clique problem and has been proven to be NP-complete \cite{DBLP:books/fm/GareyJ79}. Our goal is to prove the NP-hardness of \emph{CCS} problem by reducing a known NP-complete problem (in particular the $k$-clique problem) to our \emph{CCS} problem.

$\Rightarrow: $ Let $S$ be a $k$-clique subgraph in graph $G$ and $vol(S)<|E(G)|$. We construct an instance of the decision \emph{CCS} problem, satisfying the following conditions: (1) $q \in S$, (2) $S$ is connected, and (3) the minimum \emph{conductance} $\phi(S)=|E(S,\bar{S})|/((k(k-1)+|E(S,\bar{S})|)$. Thus, the solution to the \emph{CCS} problem includes a $k$-clique of size $k$.

$\Leftarrow: $ Assume that we have an instance of the decision \emph{CCS}  problem satisfying the following conditions: (1) $q \in S$, (2) $S$ is connected, and (3) the minimum \emph{conductance} $\phi(S)=|E(S,\bar{S})|/vol(S)$, where $vol(S)=\sum_{u \in S}d_{S}(u)+|E(S,\bar{S})|$. And because $\phi(S)= 1/(1+\sum_{u \in S}d_{S}(u)/|E(S,\bar{S})|)$, when $|E(S,\bar{S})|$ is minimized and $\sum_{u \in S}d_{S}(u)$ is maximized, i.e., $\sum_{u \in S}d_{S}(u)=k(k-1)$, the minimum \emph{conductance} $\phi(S)=|E(S,\bar{S})|/(k(k-1)+|E(S,\bar{S})|)$. Obviously, $S$ is an instance of the $k$-clique problem containing $k$ vertices.

Thus, $S\subseteq G$ is a $k$-clique iff $S$ is the solution to the decision \emph{CCS} problem. This proves that \emph{CCS} is NP-hard.
\end{proof}

Due to the NP-hardness of our problem \emph{CCS}, it is infeasible to identify the exact solution even on small graphs. Thus, we propose efficient heuristic methods for \emph{CCS} in next sections.

\begin{figure}[t]
\centering	\includegraphics[width=0.8\columnwidth]{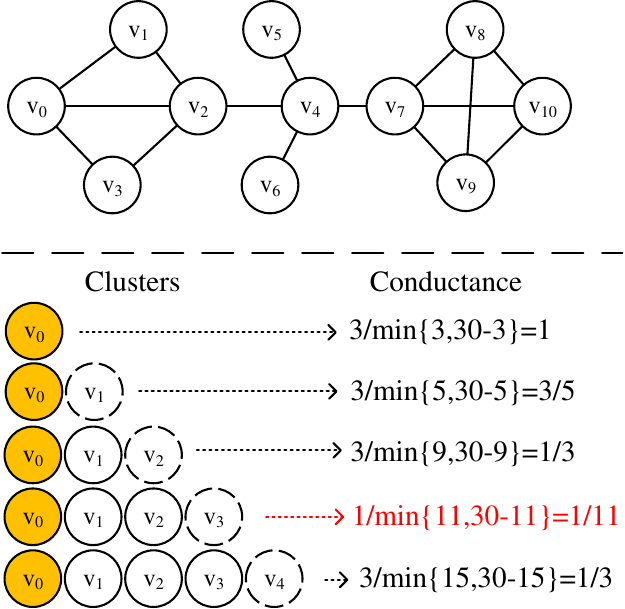} 
\caption{Illustration of conductance for a graph with 11 vertices and 15 edges. The cluster $\{v_0, v_1,v_2,v_3\}$ has the smallest conductance.}
\label{fig:problem_statement} 
\vspace*{-0.6cm}
\end{figure}

%% file: algorithm.tex
\section{A Basic Algorithm for \textit{CCS}} \label{sec:gr}
In this section, we first introduce existing SOTA solutions and conduct an in-depth analysis of their limitations. Then, we devise our basic algorithm \textit{PPRCS} and discuss its potential shortcomings.

\begin{table*}[t!]
\centering
\caption{A comparison of SOTA algorithms. $\times$ means that the corresponding method does not satisfy the corresponding condition, whereas $\checkmark$ means that the condition is met. $k$ refers to the parameter $k$ that the method requires to be predetermined.} \vspace{-0.3cm}
\scalebox{1}{
\begin{tabular}{c|c|c|c|c|c}
\toprule
\multicolumn{1}{c|}{Methods} &
\multicolumn{1}{c|}{Cohesion and Exosparcity}&
\multicolumn{1}{c|}{$k$}&
\multicolumn{1}{c|}{Seed Inclusion}& \multicolumn{1}{c|}{Connectivity Guarantee} & \multicolumn{1}{c}{Objective Function}\\
\midrule
\textit{CSM} \cite{DBLP:conf/sigmod/CuiXWW14,csm1} & $\times$ & $\checkmark$&$\checkmark$ & $\checkmark$ & Cohesive Subgraph-based   \\
\textit{TCP} \cite{DBLP:conf/sigmod/HuangCQTY14,truss1,truss2} &$\times$ & $\checkmark$ &$\checkmark$ & $\checkmark$ & Cohesive Subgraph-based\\
\midrule
\textit{SM} \cite{DBLP:conf/webi/LuoWP06} & $\checkmark$& $\times$ &$\times$ &  $\checkmark$ & Modularity-based \\
\midrule
\textit{PPR\&NIBBLE} \cite{DBLP:conf/focs/AndersenCL06,ppr1} & $\checkmark$& $\times$ &$\times$ &  $\times$ & Conductance-based\\
\textit{HK\_Relax} \cite{DBLP:conf/kdd/KlosterG14,hk1,hk2} & $\checkmark$& $\times$ &$\times$ &  $\times$ & Conductance-based\\
\textit{PCon\_de} \cite{DBLP:conf/aaai/LinLJ23} &  $\checkmark$ &$\times$ &$\times$ & $\times$ & Conductance-based \\
\textit{ASC} \cite{trevisan2017lecture} &  $\checkmark$ &$\times$&$\times$ & $\times$ & Conductance-based \\
\midrule
\textit{PPRCS} (this paper)& $\checkmark$& $\times$ &$\checkmark$ & $\checkmark$ & Conductance-based\\ 
\textit{SCCS} (this paper)& $\checkmark$& $\times$ &$\checkmark$ & $\checkmark$ & Conductance-based \\
\bottomrule	
\end{tabular}} 
\label{tab:comparison}
\end{table*}


\subsection{Existing SOTA Solutions} \label{subsec:sota}
Existing solutions typically undergo three stages. Initially, the score function for each vertex is computed. Subsequently, vertices are iteratively added or removed based on their scores. Finally, a locally optimal cluster is generated during the above iteration. For instance, the \textit{CSM} model is proposed  to maximize the minimum degree \cite{DBLP:conf/sigmod/CuiXWW14,csm1}. At each iteration, the vertex $v$ with the maximum $d_{C\cup \{v\}} (v)$ ($C$ is the vertex set that has been visited so far) is chosen to be added. The \textit{TCP} model is designed to identify more tightly connected clusters in conjunction with the $k$-truss model \cite{DBLP:conf/sigmod/HuangCQTY14,truss1,truss2}. For every vertex $v\in V$, its neighbor vertex $u\in N(v)$ is  arranged based on the edge trussness $\tau (e(u, v))$.  Luo et al. \cite{DBLP:conf/webi/LuoWP06} introduced the subgraph modularity $\textit{SM}$ model as $M=ind(C)/outd(C)$, where $ind(C)$ is the number of edges within $C$ and $outd(C)$ is the number of edges from $C$ connecting to $\bar{C}$. They encompassed both the addition and deletion phases. During the addition phase, vertices that maximize $M$ are sequentially included, whereas in the deletion phase, vertices that weaken $M$ are sequentially eliminated. Andersen et al. \cite{DBLP:conf/focs/AndersenCL06} devised the \textit{PPR\&NIBBLE} model that uses Personalized PageRank (PPR) to sort vertices and then executed a sweep cut procedure to obtain the local optimal conductance. The \textit{HK-Relax} is a diffuse model based on the score derived from the heat kernel random walk \cite{DBLP:conf/kdd/KlosterG14,hk1,hk2}. Lin et al. \cite{DBLP:conf/aaai/LinLJ23} proposed the \textit{PCon\_de} model that sorts vertices in ascending order of degree ratio and sequentially removed those with the smallest degree ratio. Trecisan et al. \cite{trevisan2017lecture} proposed the \textit{ASC} model to identify a cluster by approximating the eigenvectors associated with the second smallest eigenvalue $\lambda_2$ of the corresponding graph Laplacian matrix.


Unfortunately, though these solutions effectively tackle their respective challenges, they still exhibit several limitations in practical application (Table \ref{tab:comparison}). Firstly, cohesive subgraph-based clustering algorithms (e.g., \textit{CSM} \cite{DBLP:conf/sigmod/CuiXWW14}, \textit{TCP} \cite{DBLP:conf/sigmod/HuangCQTY14}) predominantly emphasize intra-community cohesion at the expense of neglecting their separation from the external network. On top of that, for large and intricate graphs, these algorithms demand significant computational costs to yield qualified solutions, as stated in our empirical results (Section \ref{sec:experiments}). Secondly, for the modularity-based clustering algorithms (e.g., \textit{SM} \cite{DBLP:conf/webi/LuoWP06}) and conductance-based clustering algorithms (e.g., \textit{PPR\&NIBBLE} \cite{DBLP:conf/focs/AndersenCL06},
\textit{HK\_Relax} \cite{DBLP:conf/kdd/KlosterG14},
\textit{PCon\_de} \cite{DBLP:conf/aaai/LinLJ23},
\textit{ASC} \cite{trevisan2017lecture}), they cannot ensure that resultant clusters maintain connectivity or include the query vertex $q$ \cite{DBLP:journals/kbs/DingZY18,DBLP:conf/icml/ZhuLM13}, rendering them unsuitable for our proposed \textit{CCS} problem. Moreover, from a practical standpoint, they may yield communities that are much larger than the actual community size, as stated in our empirical results (Section \ref{sec:experiments}).

\begin{algorithm}[t] 
	\caption{\textit{ PPRCS ($G$, $q$, $\alpha$, $r_{max}$)}}  \label{algor:CSPPR}
	\begin{flushleft}
		\hspace*{0.02in} {\bf Input:} 
		graph $G(V,E)$; query vertex $q$; two parameters $\alpha$ and $r_{max}$\\
		\hspace*{0.02in} {\bf Output:} a  community $S$
	\end{flushleft}
	\begin{algorithmic}[1]
\State $\hat{\pi}$ $\leftarrow$  \Call{Forward\_Push}{$G,q,\alpha,r_{max}$}
		\State $\hat{\pi}(q) \leftarrow \infty$; $y \leftarrow \hat{\pi} D^{-1}$  
\State $y_i \leftarrow $ to be index of $y$ with $i$th largest value
\State $S \leftarrow \arg\min \phi(S_i)$ with $S_i=\{y_1,y_2,...y_i\}$ and $S_i$ is connected
\State \textbf{return} $S$
		\Function{Forward\_Push}{$G,q,\alpha,r_{max}$}
		\State $\hat{\pi}(t) \leftarrow 0$ for all $t \in V$ 
		\State $r(q) \leftarrow 1$;  $r(t) \leftarrow 0$ for all $t \neq q$
		\While{$\exists t $ such that $r(t)\geq r_{max} \cdot d_t$} 		
		\For {each $u \in N(t)$} 
		\State $r(u) \leftarrow r(u)+(1-\alpha)r(t)/d_t$ 
		\EndFor
		\State  $\hat{\pi}(t) \leftarrow \hat{\pi}(t) + \alpha r(t)$; $r(t) \leftarrow 0$ 
		\EndWhile
		\State \Return $\hat{\pi}$
				\EndFunction
	\end{algorithmic}
\end{algorithm}

\subsection{Our Basic Algorithm: \textit{PPRCS}} \label{subsec:The CSPPR Algorithm}
To address the \emph{CCS} problem, we introduce an algorithm called \underline{P}ersonalized \underline{P}age\underline{R}ank-based \underline{C}ommunity \underline{S}earch (\emph{PPRCS}) as shown in Algorithm \ref{algor:CSPPR}. The core idea of \emph{PPRCS} is to first compute the proximity of each vertex, and then use a greedy scan to find the community with the minimum \emph{conductance}. During the scan, we ensure that the extracted community remains connected, and set the importance of the query vertex to infinity to guarantee that the community  always contains the query vertex.

Specifically, \emph{PPRCS} first calculates the PPR score $\hat{\pi}(t)$ for each vertex (Lines 1, 6-13). The $\hat{\pi}(t)$ of a vertex $t \in V$ indicates the probability of reaching the vertex $t$ from the query vertex $q$ through a random walk with $\alpha$-decay (Lines 10-12). A higher $\alpha$ implies a greater likelihood for a vertex to remain in the neighborhood of the query vertex, whereas a smaller $\alpha$ means the vertex is more influenced by its neighbors. Thus, the $\hat{\pi}(t)$ of $t$ relative to $q$ inherently serves as a measure of proximity between $q$ and $t$. Then, by sorting these vertices based on their $\hat{\pi}(t)$ values into the vector $y$ (Lines 2-3), we prioritize processing the top-ranked vertices. Finally, we select the first connected subgraph with the minimum \textit{conductance} as target community $S$ (Lines 4-5).

\begin{example}
Reconsider Figure \ref{fig:problem_statement}, we assume that the query vertex $q=v_7$, $\alpha = 0.10$ and $r_{max} = 0.0001$. By executing the Forward\_Push (i.e., Lines 1, 6-13 of Algorithm \ref{algor:CSPPR}), we obtain the PPR vector $\hat{\pi}$ is as follows: $\{v_7:0.9028, v_8:0.0241, v_9:0.0241, v_{10}:0.0240, v_4:0.0224, v_2:0.0005, v_5:0.0005, v_6:0.0005, v_0:0, v_1:0, v_3:0\}$. The descending order of $y_{i}$ (i.e., Lines 2-3 of Algorithm \ref{algor:CSPPR}) is as follows: $\{v_7: 0.2257, v_{9}: 0.0080, v_8: 0.0080, v_{10}: 0.0080, v_4: 0.0056, v_6: 0.0005, v_5: 0.0005, v_2: 0.0001, v_0:0, v_1:0, v_3:0\}$, then $S = \{v_8, v_9, v_{10}, v_7\}$ is the first scanned connected community with the smallest  $\phi(S) = 0.076$ and we return $S$ as our answer (Lines 4-5 of Algorithm \ref{algor:CSPPR}).
\end{example}

\stitle{Complexity Analysis of \textit{PPRCS}.} In the Forward\_Push function, the number of iterations is influenced by $\alpha$ and $r_{max}$, where $\alpha$ is the jump probability and $r_{max}$ is the error threshold. Hence, the total time complexity of this function is $O(1/(r_{max} \times \alpha))$ \cite{DBLP:conf/focs/AndersenCL06}. In the main function, $\hat{\pi}(q)$ is first set to infinity, which incurs a constant time complexity of $O(1)$. Then, traversing each vertex in $\hat{\pi}$ to compute $y = \hat{\pi} D^{-1}$, where $D$ is the degree matrix, has a time complexity of $O(n)$. Subsequently, sorting $y_{i}$ using the built-in function $sorted()$ takes $O(n \log n)$ time. Lastly, vertices are sequentially added to the set $S$ according to the order of $y_{i}$, and the \textit{conductance} is calculated after each addition. This step also has a time complexity of $O(n)$. Thus, the worst-case time complexity of \emph{PPRCS} is  $O(1/(r_{max}\times \alpha) + n \log n)$. \emph{PPRCS} needs extra $O(n)$ to store $\hat{\pi}$, $r$ and $y$ to obtain the result, thus the worst-case space complexity of \emph{PPRCS} is $O(n+m)$.

\stitle{Limitations of \textit{PPRCS}.} Although \textit{PPRCS} proves effective, its heavy reliance on parameters like $\alpha$ and $r_{max}$ often leads to instability concerning clustering quality and computation time \cite{DBLP:conf/icml/ZhuLM13,DBLP:conf/kdd/KlosterG14}. Moreover, the resultant cluster sizes may not accurately reflect the ground-true community sizes, as stated in our empirical results (Section \ref{sec:experiments}). To overcome these drawbacks, we propose the following algorithm with several optimizations to improve performance.

\section{An Advanced Algorithm for \textit{CCS}} \label{sec:LCS}
In this section, we start by sampling a subgraph $S_{G}$ from the original graph $G$ and then perform three-stage optimization on $S_{G}$. Finally, we present our advanced solution \emph{SCCS}.

\vspace{-0.2cm}
\subsection{Sampling Subgraph $S_{G}$}  \label{subsec:Sampling}

Since performing the search directly in the original large graph can lead to prohibitively high time\&space overheads, sampling-based optimization strategies have been widely studied \cite{DBLP:journals/tkdd/HeSBH19,DBLP:journals/tbd/WangYBH24}. Communities are typically composed of vertices that are close to each other, indicating the existence of short paths between vertices within a community \cite{macqueen2001community}. Inspired by these, limiting the sampling distance helps us focus more on the neighboring vertices closely related to the query vertex and avoid the exploration of regions far from the query vertex, reducing significantly computational costs and improving the scalability of the algorithm.

To address the distance issue, we transform it into the computation of the depth of the tree. Assuming the query vertex $q$ is located in the connected component $C_{q}$ and serves as the root node with a depth of 0. For any vertex $v_{i}$ in the connected component $C_{q}$, it resides at different depths as a child node.

\begin{definition}
[Depth $d_G(v_i, q)$] \label{thm:depth}
Given any vertex $v_i$, the minimum depth from $v_i$ to $q$ is expressed as
\begin{equation}
d_{G}(v_{i},q) = \begin{cases}
\text{depth}(v_{i},q), & \text{if } v_{i}\in C_q \\
inf, & \text{otherwise }
\end{cases}
\end{equation}
\end{definition}

If vertex $v_{i}$ and query vertex $q$ do not belong to the same connected component, the distance between them is defined as $inf$ to reflect their inaccessibility.

We aim to identify a local structure $S_{G}$ composed of vertices $v_{i}$ that satisfy $d_{G}(v_{i},q) \leq dp$. We need to consider two potential scenarios. In networks with low degrees, $S_{G}$ may expand insufficiently, so we continue accessing vertices at depth $dp = dp + 1$ until the lower limit $l$ is met. Conversely, in dense networks, to prevent $S_{G}$ from expanding excessively, we halt the search when the number of vertices in the vertex set $V_{S_{G}}$ reaches the predetermined upper limit $h$. We adopt the idea of breadth-first search in Algorithm \ref{algor:Sampling subgraph}. Specifically, we first initialize $V_{S_{G}}$ and a queue with the query vertex $q$ and its depth of 0 (Line 1). When the queue is not empty, we pop vertex $u$ and its depth $d_{G}(u, q)$ from the queue, traverse the neighbors $v_i\in N_{\bar{S}_{G}}(u)$ of $u$ and update $S_G$ and the queue, where $N_{\bar{S}_{G}}(u)=N(u) - N_{V_{S_{G}}}(u)$ (Lines 2-7). If the number of subgraph vertices reaches the minimum threshold $l$ and the current depth is greater than the maximum depth $dp$, access is interrupted early (Lines 8-9). Otherwise, continue traversing each neighbor $v_i\in N_{\bar{S}_{G}}(u)$ of $u$ until the number of subgraph vertices exceeds the maximum threshold $h$, then the traversal ends (Lines 10-14). Note that the "break" in lines 7 and 14 interrupts the outer loop. Finally, the sampled subgraph $S_G$ consisting of $V_{S_G}$ and $E_{S_G}$ is returned (Lines 15-16).

\begin{algorithm}[t]
\caption{\textit{ Sampling ($G$, $q$, $dp$, $l$, $h$)}}  \label{algor:Sampling subgraph}
\begin{flushleft}
\hspace*{0.02in} {\bf Input:} 
An undirected graph $G(V,E)$; query vertex $q$; depth $dp$; minimum threshold $l$; maximum threshold $h$\\
\hspace*{0.02in} {\bf Output:} 
A sampling subgraph $S_{G}$
\end{flushleft}
\begin{algorithmic}[1]
\State Initialize: $V_{S_{G}} \leftarrow {q}$, queue $\leftarrow [(q, 0)]$
\While{queue}
\State $u, d_{G}(u,q) \leftarrow$ queue.pop()
\If{$|V_{S_{G}}| < l$}
\For{each $v_{i}\in N_{\bar{S}_{G}}(u)$}
\State $V_{S_{G}} \leftarrow V_{S_{G}} \cup \{v_{i}\}$; $E_{S_{G}} \leftarrow E_{S_{G}} \cup {(u, v_{i})}$
\State queue.push$((v_{i}, d_G(u,q) + 1))$
\If{$|V_{S_{G}}| \geq l$ and $d_{G}(u, q) > dp$}
\State break
\EndIf
\EndFor
\EndIf
\For{each $v_{i}\in N_{\bar{S}_{G}}(u)$}
\State $V_{S_{G}} \leftarrow V_{S_{G}} \cup \{v_{i}\}$; $E_{S_{G}} \leftarrow E_{S_{G}} \cup {(u, v_{i})}$
\State queue.push$((v_{i}, d_{G}(u, q) + 1))$
\If{$|V_{S_{G}}|> h$}
\State break
\EndIf
\EndFor
\EndWhile
\State $S_{G} \leftarrow V_{S_{G}}$, $E_{S_{G}}$
\State \textbf{return} $S_{G}$
\end{algorithmic}
\end{algorithm}

\subsection{Three-stage Community Optimization} \label{subsec:Community Extraction}

In this subsection, we extract a community from a sampled subgraph $S_{G}$ with the goal of identifying the connected community $S^{*}$ that contains the query vertex $q$ and has the highest quality. Our optimization process consists of three stages: the initial community stage, the expansion stage, and the verification stage.

\subsubsection{The initial community stage} \label{subsec:initialPhase}

For community search, the choice of initial communities plays a crucial role in the performance of algorithms. This paper proposes a seed expansion method based on the largest clique, in which the largest clique containing the query vertex $q$ is selected as the initial community. The maximum $k$-clique is considered as a potential method for initial community selection, as it can identify the most tightly connected substructures in the network, i.e., the maximum complete subgraph consisting of $k$ vertices. This approach emphasizes the high level of connectivity among vertices within communities and aims to aggregate vertices with similar features or functions. By selecting the maximum $k$-clique as the initial community, it can  preserves the local structure of the network and provides a good starting point with potential community structure for the community search process.

It is worth noting that the initial community consists of all maximal cliques (with vertex count greater than or equal to 4) containing the specified vertex $q$ and only returns maximal cliques that do not be contained by other cliques. For instance, in Figure \ref{fig:max_cliques}, if $q = v_{0}$ and $v_{0}\in 4-clique \subseteq 5-clique\subseteq 6-clique$, the initial community should be $6-clique = \{v_0, v_1, v_2, v_3, v_4, v_5\}$ containing $v_0$. 
However, if $q = v_{5}$ and it belongs to communities $S_{1}$ and $S_{2}$, where $S_{1}$ and $S_{2}$ intersect at $v_{5}$ (i.e, $4-clique \not\subseteq 6-clique$), which shows that both $S_{1}$ and $S_{2}$ are initial communities containing $v_{5}$. In this way, we consider multiple different maximal cliques containing $q$ to ensure that all relevant community structures are taken into account. Thus, we introduce a new strategy for initializing communities, referred to as "$initialCom$" that requires $clique_{i} \not\subseteq clique_{j}$, and $clique_{i} \cap clique_{j} = I$, where $i$ and $j$ are the $i$-th and $j$-th largest cliques in $initialCom$ respectively, $j>i>0$ and $q\in I$. The $initialCom$ strategy is defined as follows:

\begin{align}
initialCom = \{clique \mid & \, clique_i \not\subseteq clique_j, \nonumber\\
& clique_i \cap clique_j = I, \nonumber\\
& \forall j > i > 0, q \in I\}
\end{align}

\begin{figure}[t!]
\centering
\includegraphics [width=0.45\textwidth] {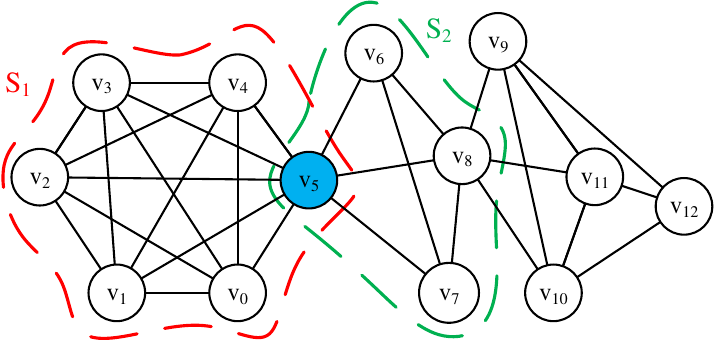}	
\caption{An example of the maximum $k$-clique-communities when $k\geq 4$. $S_1$ and $S_2$ share the vertex $v_5$.} 
\label{fig:max_cliques}
\end{figure}

\subsubsection{The expansion stage} \label{subsec:expansionPhase}

While the maximum $k$-clique ensures tight connections within a community, it focuses solely on local internal connectivity and overlooks the interaction between the community and the entire network. In community search, apart from the tightness of internal structure, it is imperative to take the relationship between the community and the entire network into consideration. Therefore, we intend to expand the community around $S\in initialCom$: Select a vertex that maximizes community quality from the neighbors of $S$ to join the community, and stop expanding the community until there are no neighbors or the joining of neighbors makes the community quality no longer increase. But, adding only one vertex at a time may prematurely lead the expansion process to local optima, yielding a suboptimal result.

To overcome this limitation, we consider adding multiple vertices in Algorithm \ref{algor:expansion}. Initially, we select the vertex $v_{top}\in N_{\bar{S'}}(S')$ that maximizes the community quality and temporarily store it in $A_i$, where $N_{\bar{S'}}(u)=N(u) - N_{S'}(u)$ and $N_{\bar{S'}}(S') = \sum_{u\in S'}{N_{\bar{S'}}(u)}$(Lines 4-5). If adding $v_{top}$ causes the quality of the community $S'$ to decrease, we consider adding the next $v_{top}$ to $S'$ as well (Lines 6-8). If the quality gain $\Delta f(S' \cup \{v_{top}\})_{max}\geq 0$, we update $S$ to $S \cup A_i$ (Line 9). This process is repeated based on this idea until the neighbor set $N_{\bar{S'}}(S')$ is empty or the number of vertices added at one time reaches the specified threshold $count$ without improving the community quality (Line 3). Finally, the expanded community $S$ is returned (Line 10). Due to the observation that $S_G \ll G$, it is often possible to simplify the denominator $\min\{vol(S),2m-vol(S)\}$ of $\phi(S)$ to $vol(S)$. Consequently, the \textit{conductance} $\phi(S)$ can be simplified to the \textit{subgraph conductance} $\phi_{S}(S)$.

\begin{definition}
[Subgraph conductance $\phi_{S}(S)$] \label{thm:phi}
Given any sampled subgraph $S_G\ll G$ and a vertex subset $S\subseteq S_G$, the \textit{conductance} of $S$ is expressed as
\begin{equation} \phi_{S}(S)=\frac{|E(S,\bar{S})|}{vol(S)}
\end{equation}
\end{definition}

\begin{definition}
[Quality score $f(S)$] \label{def:f(S)}
Given any vertex set $S\subseteq S_G$, the quality score of $S$ is expressed as \begin{equation} f(S) = \frac{D_{in}(S,S)}{D_{in}(S,S) + E_{out}(S,\bar{S})}
\end{equation}
\end{definition}

\begin{proof}
Suppose we use $links(,)$ to represent the number of connected edges. By Definition \ref{thm:phi}, we have $\phi_{S}(S)=\frac{|E(S,\bar{S})|}{vol(S)} = 1- \frac{D_{in}(S,S)}{D_{in}(S,S)+E_{out}(S,\bar{S})}$, where $E_{out}(S,\bar{S})=links(S,\bar{S})$ and $D_{in}(S,S) = 2 \times links(S,S)$.
\end{proof}

\begin{algorithm}[t]
\caption{\textit{ Expansion ($S$, $q$, $count$)}}  
\label{algor:expansion}
\begin{flushleft}
\hspace*{0.02in} {\bf Input:} 
An initial community $S$; query vertex $q$; extend the upper limit of the number of vertices at once $count$\\
\hspace*{0.02in} {\bf Output:} 
An extended community $S$
\end{flushleft}
\begin{algorithmic}[1]
\State $S'\leftarrow S$; $f(S) \leftarrow D_{in}(S, S), E_{out}(S, \bar{S})$ 
\State $f(S')\leftarrow f(S); A_i\leftarrow \emptyset; i\leftarrow 0$
\While {$N_{\bar{S'}}(S')\neq \emptyset$ and $count > i$} 
\State find the $v_{top}$ with $\Delta f(S'\cup \{v_{top}\})_{max}$
\State $A_{i+1}\leftarrow A_{i}\cup \{v_{top}\}$; $i\leftarrow i+1$
\If {$\Delta f(S'\cup \{v_{top}\})_{max} < 0$}
\State $S'\leftarrow S'\cup \{v_{top}\}$
\State \textbf{continue}
\EndIf
\State $S\leftarrow S\cup A_{i}$; $i\leftarrow 0$
\EndWhile
\State \textbf{return} $S$
\end{algorithmic}
\end{algorithm}

Obviously, the higher the value of $f(S)$, the higher the quality score of the community $S$. When we add adjacent vertices to $S$, the quality score of $S$ will change. Therefore, we use $\Delta f(S\cup A_{i})$ to represent the gain change in quality after $S$ is expanded to $S\cup A_{i}$.

\begin{definition}
[Quality gain $\Delta f(S\cup A_i)$] \label{thm:gain++}
Given any vertex set $S\subseteq S_G$, the quality gain obtained after expanding $S$ to $S\cup A_i$ is expressed as \begin{equation}
\Delta f(S\cup A_{i}) = f(S\cup A_{i}) - f(S)
\label{eq:gain++}
\end{equation}
where $A_{i}$ represents a set of adjacent vertices of size $i\in [1,|N_{\bar{S}}(S)|]$.
\end{definition}

From Eq. \ref{eq:gain++}, the gain $\Delta f(S\cup A_{i})$ requires recalculating the connections related to $S \cup A_{i}$ each time. To ensure a fast calculation of the changes in quality gain $\Delta f(S\cup A_{i})$ re-induced by $f(S\cup A_i) - f(S)$, we only calculate the connections related to $A_{i}$, which can be calculated in $O(1)$ time. Let $D_{in}(S, A_{i})$ represent the internal volume between $S$ and $A_{i}$:
\begin{equation}
D_{in}(S, A_{i})= 2 \times (links(S, A_{i}) + links(A_{i},A_{i}))
\label{eq:D_in}
\end{equation}

$E_{out}(S, A_{i})$ and $E_{out}(\bar{S}, A_{i})$ represent the connections between $S$ and $A_{i}$ and between $\bar{S}$ and $A_{i}$, respectively 
\begin{flalign}
E_{out}(S, A_{i}) &= links(S, A_{i})
\label{eq:E_out1}
\\
E_{out}(\bar{S}, A_{i}) &= links(\bar{S}, A_{i})\label{eq:E_out2}
\end{flalign}

\begin{theorem}
[Quality of $f(S\cup A_i)$] 
Given any vertex set $S\subseteq S_G$, the quality of $S\cup A_i$ is expressed as
{\scriptsize
\begin{equation}
f(S\cup A_{i}) = \frac{D_{in}(S, S) + D_{in}(S, A_{i})}{D_{in}(S,S) + E_{out}(S,\bar{S}) + D_{in}(S,A_{i}) + E_{out}(A_{i})}
\label{eq:S+A}
\end{equation}
}
where $E_{out}(A_{i}) = E_{out}(\bar{S},A_{i}) - E_{out}(S,A_{i})$.
\end{theorem}

\begin{proof}
According to the proof of Definition \ref{def:f(S)},
we have $E_{out}(S,\bar{S})=links(S,\bar{S})$ and $D_{in}(S,S)=2\times links(S,S)$. And from Eq.\ref{eq:D_in}, the internal volume of $S\cup A_{i}$ is
{\scriptsize
\begin{flalign}
   D_{in}(S\cup A_{i}) &= 2 \times links(S\cup A_{i}, S\cup A_{i}) \nonumber\\&= 2 \times links(S, S) + 2\times links(A_{i}, A_{i}) + 2 \times links(S, A_{i}) \nonumber\\&= D_{in} (S,S) + D_{in}(S, A_{i})
\label{eq:D1}
\end{flalign}
}

From Eq.\ref{eq:E_out1}-\ref{eq:E_out2}, the number of cut edges of $S\cup A_{i}$ is
{\scriptsize
\begin{flalign}
    E_{out}(S\cup A_{i}) &= links(S\cup A_{i}, \overline{S\cup A_{i}}) \nonumber\\&= links(S,\bar{S}) +links(\bar{S}, A_{i})-links(S, A_{i}) \nonumber\\& = E_{out}(S,\bar{S}) + E_{out}(\bar{S},A_{i}) - E_{out}(S,A_{i})
\label{eq:E1}
\end{flalign}
}
Therefore, from Eq.\ref{eq:D1}-\ref{eq:E1}, we get Eq. \ref{eq:S+A}  = Eq. \ref{eq:S+Ai}
{\scriptsize
\begin{equation}
   f(S\cup A_{i}) =\frac{ D_{in}(S\cup A_{i})}{ D_{in}(S\cup A_{i}) + E_{out}(S\cup A_{i})}
\label{eq:S+Ai}
\end{equation}
}
\end{proof}

Hence, to calculate the gain $\Delta f(S\cup A_{i})$, we only need to update $D_{in}(S, A_{i})$, $E_{out}(\bar{S}, A_{i})$, and $E_{out}(S, A_{i})$.

\begin{example}
Considering the example of Figure \ref{fig:max_cliques}, Figure \ref{fig:expansion} shows the relevant steps of the expansion. Firstly, we initialize the community as $S' = S = \{v_5, v_6, v_7, v_8\}$, thus $f(S') = f(S) = \frac{2 \times 6}{2 \times 6 + 8} = \frac{12}{20} = 0.6$. Next, we consider $f(S' \cup \{v_{top}\})$ after expanding the neighbors of 
$S'$. We obtain $f(S' \cup \{v_p\}) = \frac{2 \times 7}{2 \times 7 + 11} = \frac{14}{25} = 0.560$ and $f(S' \cup \{v_q\}) = \frac{2 \times 7}{2 \times 7 + 10} = \frac{14}{24} = 0.583$, where $p \in [0,4]$ and $q \in [9,11]$. We select any vertex (such as $v_9$) from the set with maximum $f(S' \cup \{v_q\})$ to add to $S'$ and update $A_1=\{v_9\}$. Since $\Delta f(S' \cup \{v_9\})<0$, we consider adding the next vertex that maximizes $\Delta f(S'\cup \{v_{top}\})$ to $S'$. Similarly, we get $f(S'\cup \{v_p\}) = \frac{2 \times 8}{2 \times 8 + 13} = \frac{16}{29} = 0.551$, $f(S'\cup \{v_{10}\}) = f(S'\cup \{v_{11}\}) = \frac{2 \times 9}{2 \times 9 + 10} = \frac{18}{28} = 0.642 $ and $f(S'\cup \{v_{12}\})=\frac{2\times 8}{2 \times 8 + 11} = \frac{16}{27} = 0.592$. When adding $v_{10}$ to $S'$, we have $f(S'\cup \{v_{10}\}) = 0.642 > f(S)=0.6$ and $A_2=\{v_9,v_{10}\}$. Thus, $S$ can be updated as $S\cup A_2$.
\end{example}

\begin{figure}[t!]
\centering
\includegraphics [width=0.45\textwidth] {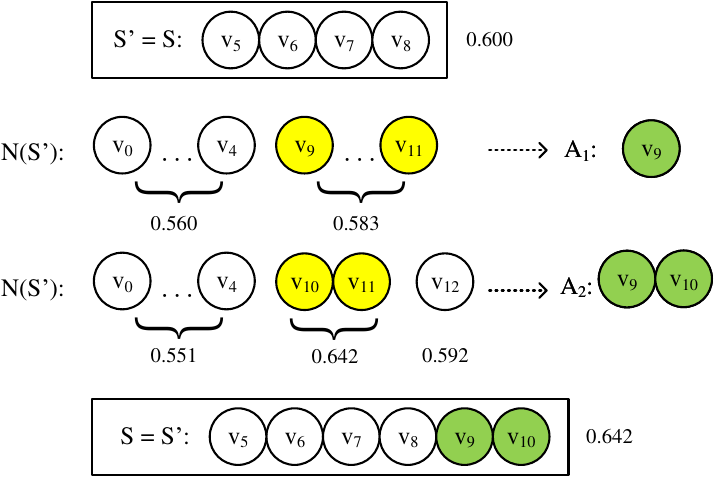}
\caption{An example of the expansion step. Given the initial community $S$, where $D_{in}(S, S) = 2 \times 6 = 12$ and $E_{out}(S, \bar{S}) = 8 $. When adding $A_2 = \{v_{9}, v_{10}\}$, we have $D_{in}(S, A_2) = 2 \times 1 + 2 \times 2 = 2 + 4 = 6$, $E_{out}(\bar{S}, A_2) = 5$, and $E_{out}(S, A_2) = 3$. Therefore, $f(S \cup A_2) = \frac{12 + 6}{12 + 8 + 6 + 5 - 3} = \frac{18}{28} = 0.642$.} 
\label{fig:expansion}
\end{figure}

\subsubsection{The verification stage} \label{subsec:checkingPhase}

Although we expand the community in the direction of improving $f(S)$, the non-monotonicity of \textit{conductance} means that previously added vertices may also adversely affect $f(S)$. Therefore, we introduce a verification stage (Algorithm  \ref{algor:checking}) to iteratively remove the vertex that is unfavorable to $f(S)$, which also helps the expansion stage to consider more vertices. Specifically, we let $B_{j}$ represent the set of boundary vertices belonging to $S$ that does not contain $q$ and has at least one edge in $\bar{S}$, where $j$ is the size of $B$ (Line 1). 
We define a flag $changed$ to indicate whether the community has been modified and start flagging it as $False$ (Line 2). Then we traverse $B_j$ and iteratively remove the vertex that makes $f(S-\{v_i\})$ increase without destroying community connectivity (Lines 3-7). If such a removal exists, we change the flag to $True$ (Line 8).

\begin{algorithm}[t]
\caption{\textit{ Verification ($S$, $q$)}} \label{algor:checking}
\begin{flushleft}
\hspace*{0.02in}{\bf Input:} 
An extended community $S$; query vertex $q$\\
\hspace*{0.02in}{\bf Output:} 
A pruned community $S$
\end{flushleft}
\begin{algorithmic}[1]
\State $B_j\leftarrow \{v_i\in S,v_i\neq q|(v_i,u)\in E,u\in \bar{S}\}$
\State changed $\leftarrow$ False
\For{$v_i$ in $B_j$}
\State $S''\leftarrow S-\{v_i\}$
\If{$\Delta f(S'')>0$ and $S''$ is connected}
\State $B_{j-1} \leftarrow B_{j} - \{v_{i}\}$; $j \leftarrow j-1$
\State $S \leftarrow S - \{v_{i}\}$; update $f(S)$
\State changed $\leftarrow$ True
\EndIf
\EndFor
\State \textbf{return} $S$
\end{algorithmic}
\end{algorithm}

\begin{definition}
[Quality gain $\Delta f(S- B_j)$] \label{thm:gain--}
Given any vertex set $S\subseteq S_G$, the quality gain obtained after shrinking $S$ to $S-B_j$ is expressed as 
\begin{equation}
    \Delta f(S-B_{j}) =f(S-B_{j}) - f(S)
\label{eq:gain--}
\end{equation}
\end{definition}

Similar to Eq. \ref{eq:gain++}, $\Delta f(S- B_{j})$ only needs to calculate the connections related to $B_{j}$, that is, $D_{in}(S,B_j)$, $E_{out}(S,B_j)$ and $E_{out}(\bar{S},B_j)$.

\begin{theorem}
[Quality of $f(S- B_j)$] 
Given any vertex set $S\subseteq S_G$, the quality of $S - B_j$ is expressed as \begin{equation}
\scriptsize
    f(S- B_{j}) = \frac{D_{in}(S,S)-D_{in}(S,B_{j})}{D_{in}(S,S) + E_{out}(S,\bar{S}) - D_{in}(S,B_{j})- E_{out}(B_{j})} 
\label{eq:S-B}
\end{equation}
where $E_{out}(B_{j}) =E_{out}(\bar{S},B_{j}) -E_{out}(S,B_{j})$.
\end{theorem}

\begin{proof}
The internal volume of $S-B_{j}$ is as follow: 
{\scriptsize
\begin{align}
D_{in}(S-B_{j}) &= 2\times links(S-B_{j}, S-B_{j}) \nonumber\\&= 2 \times links(S, S)-2 \times links(B_{j}, B_{j}) - 2 \times links(S, B_{j}) \nonumber\\&= D_{in}(S,S) - D_{in}(S,B_{j})
\label{eq:D2}
\end{align}}
The number of cut edges of $S-B_{j}$ is
{\scriptsize
\begin{flalign}
  E_{out}(S-B_{j}) &= links(S-B_{j}, \overline{S-B_{j}}) \nonumber\\&= links(S,\bar{S}) + links(S,B_{j}) -links(\bar{S},B_{j}) \nonumber\\&= E_{out}(S,\bar{S}) + E_{out}(S,B_{j}) - E_{out}(\bar{S},B_{j})
\label{eq:E2}
\end{flalign}}
Therefore, from Eq. \ref{eq:D2}-\ref{eq:E2}, we get Eq. \ref{eq:S-B}  = Eq. \ref{eq:S-Bj}
{\scriptsize
\begin{equation}
    f(S-B_{j}) = \frac{D_{in}(S-B_{j})}{D_{in}(S-B_{j}) + E_{out}(S-B_{j})}
\label{eq:S-Bj}
\end{equation}}
\end{proof}

\begin{example}
Reconsidering Figure \ref{fig:max_cliques}, Figure \ref{fig:checking} illustrates the validation steps of the algorithm. Assuming the query vertex is $q = v_6$, the expanded community $S = \{v_5, v_6, v_7, v_8, v_9, v_{10}, v_{11}, v_{12}\}$, $f(S) = \frac{2 \times 15}{2 \times 15 + 5} = \frac{30}{35} = 0.857$ and $B_j = B_1 = \{v_5\}$. After removing $v_5$, we have $f(S - \{v_5\}) = \frac{2 \times 12}{2 \times 12 + 3} = \frac{24}{27} = 0.888 > f(S)=0.857$. Therefore, we update $S$ to $S- \{v_5\}$, resulting in $f(S)=0.888$ and $B_j=B_2=\{v_7,v_8\}$. Similarly, we traverse $B_2=\{v_7,v_8\}$. Since removing $v_8$ would break the community's connectivity, we only consider removing $v_7$. We have $f(S-\{v_7\})=\frac{2\times 10}{2\times 10+4}=\frac{20}{24}=0.833<f(S)$. Therefore, the final community is $S = \{v_6,v_7,v_8,v_9,v_{10},v_{11},v_{12}\}$, with $f(S)=0.888$.
\end{example}

\begin{figure}[t!]
\centering
\includegraphics [width=0.45\textwidth] {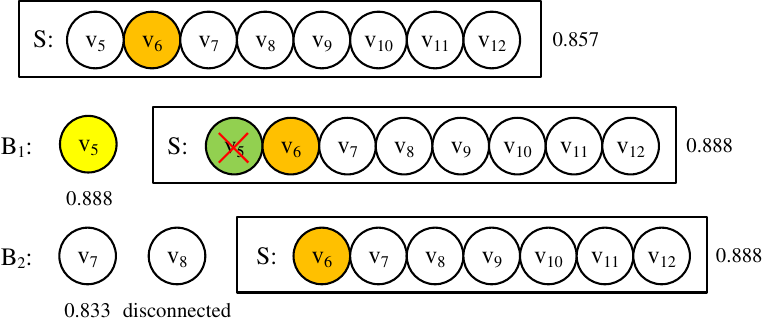}	
\caption{An example of the verification step. Let $q = v_6$, and the expanded community is $S$, where $D_{in}(S, S) = 2 \times 15 = 30$ and $E_{out}(S, \bar{S}) = 5$. $B_1 = \{v_5\}$ and after removing $v_5$, we have $D_{in}(S, B_1) = 2 \times 3 = 6$, $E_{out}(\bar{S}, B_1) = 5$, and $E_{out}(S, B_1) = 3$. Therefore, $f(S - B_1) = \frac{30 - 6}{30 + 5 - 6 - 5 + 3} = \frac{24}{27} = 0.888$.} 
\label{fig:checking}
\end{figure}

\subsection{Our Advanced Algorithm: \emph{SCCS}} 
\label{subsec:LCS}

In this subsection, we synthesize the above stages and present the overall four-stage algorithm, \underline{S}ubgraph-\underline{C}onductance-based \underline{C}ommunity \underline{S}earch (\emph{SCCS}). Specifically, Algorithm \ref{algor:SCCS} first samples the graph $G$ to avoid visiting unnecessary vertex regions, thereby accelerating the efficiency of community search (Line 1). Subsequently, the largest clique $S\in initialCom$ containing $q$ on the sampling subgraph $S_G$ is determined as the initial community to ensure cohesiveness within the community (Line 2). Then $S$ is expanded in the direction of improving the community quality $f(S)$, which considers both the community's internal and external connections (Line 4). Finally, to further improve the accuracy, we introduce a verification stage to check if there are vertices that need to be removed that are detrimental to $f(S)$ (Line 5). Repeat this expansion and verification operation multiple times until there are no vertices that can be added or deleted (Lines 6-7), and the returned community is target community $S^*$ (Lines 8-9).

\stitle{Complexity Analysis of \emph{SCCS}.} Algorithm \ref{algor:Sampling subgraph} samples $G$ through the idea of breadth-first search until the stopping condition is met to obtain the sampled subgraph $S_G$. Specifically, if $|S_{G}|\geq l$ and $d_{G}(u, q)> dp$ where $u\in S_{G}$,the sampling algorithm processes at most $l$ vertices and their associated edges. This results in a time complexity of $O(l+l(l-1)/2)$. Similarly, if $|S_{G}|> h$ during the sampling process, it processes at most $h$ vertices and associated edges. This results in a worst-case time complexity of $O(h+h(h-1)/2)$. Consequently, the time complexity of sampling can be approximated as $O(\max\{l+l(l-1)/2,h+h(h-1)/2\})$, which is further simplified to $O(h^2)$. Next, finding a maximum clique in $S_{G}$ is an NP-hard problem, hence there is no deterministic algorithm that can solve it in polynomial time. In the worst-case scenario, the time complexity of finding a maximum clique in $S_{G}$ is exponential, specifically $O(3^{h/3})$, where $h$ is the maximum number of vertices in the sampling. Then, Algorithm \ref{algor:expansion} focuses on the while loop. The time complexity of finding and removing the vertex with the maximum gain $\Delta f(S\cup \{v_{top}\})$ from the heap is $O(h \log h)$. The time complexity of updating $S$ is $O(1)$ because $count$ is a constant that does not vary with the input size. Therefore, by looping these operations $h$ times, the approximate time complexity of Algorithm \ref{algor:expansion} is $O(h^2 \log h)$. Finally, the time complexity of Algorithm \ref{algor:checking} to initialize the boundary set $B_j$ is $O(h)$. The time complexity to verify connectivity in the main loop is $O(h^2)$, and the time complexity of other operations is $O(1)$. Thus, the total time complexity of Algorithm \ref{algor:checking} is $O(jh^2)$, where $j$ is the number of boundary vertices.

\begin{algorithm}[t]
\caption{\textit{ SCCS ($G$, $q$, $dp$, $l$, $h$, $count$)}}  
\label{algor:SCCS}
\begin{flushleft}
\hspace*{0.02in} {\bf Input:} 
An undirected graph $G(V,E)$; query vertex $q$; depth $dp$; minimum threshold $l$; maximum threshold $h$; extend the upper limit of the number of vertices at once $count$\\
\hspace*{0.02in} {\bf Output:} 
A target community $S^*$
\end{flushleft}
\begin{algorithmic}[1]
\State $S_{G}\leftarrow Sampling(G,q,dp,l,h)$
\State $S\leftarrow initialCom(S_{G},q)$ 
\While {True} 
\State $S\leftarrow Expansion(S, q, count)$
\State $S\leftarrow Verification(S,q)$
\If{changed == False}
\State \textbf{break}
\EndIf
\EndWhile
\State $S^*\leftarrow S$
\State \textbf{return} $S^*$
\end{algorithmic}
\end{algorithm}

Based on the above analysis, the time complexity of \emph{SCCS} is $O(3^{h/3} + wh^2(\log h + j)$, where $w$ is the number of iterations in the outer while loop. \emph{SCCS} operates on the adjacency list representation of the graph $G$, leading to a space complexity of $O(n+m)$, where $n$ is the number of vertices and $m$ is the number of edges.

%% file: experiment.tex
\section{Experimental Evaluation}\label{sec:experiments}
In this section, we conduct thorough experiments to evaluate the performance of the proposed solutions, which are run on an Ubuntu 18.04 server with an Intel Xeon 2.50GHz CPU and 32GB of RAM. 

\subsection{Experimental Setup}

\stitle{Datasets.} 
We evaluate the solutions on six real-world datasets (http://snap.stanford.edu/) with ground-truth communities (Table \ref{tab:data}), sorted in ascending order of edges. Amazon is a network of product purchases and defines a real community for each product category. DBLP constructs a co-authorship network where authors who publish papers in certain journals or conferences form a community. Youtube, LiveJournal (LiveJ), Orkut, and Friendster are all social networks where a group of users building friendships is considered a ground-truth community. We evaluate the top 5,000 highest-quality communities for each dataset \cite{DBLP:conf/icdm/YangL12}, excluding ground-truth communities with fewer than 3 vertices. Besides, we also use the well-known NetworkX Python package \cite{nx} to generate the synthetic \emph{LFR} \cite{lancichinetti2009detecting} benchmark datasets for evaluating the performance. Note that these \emph{LFR} datasets can simulate the community properties in the real world.

\stitle{Competitors.} The following several cutting-edge competitors are complemented for comparison: \textit{CSM} \cite{DBLP:conf/sigmod/CuiXWW14,csm1}  targets to identify the subgraph containing the query vertex $q$ and having the max-min degree. \textit{TCP} \cite{DBLP:conf/sigmod/HuangCQTY14,truss1,truss2} models higher-order truss communities based on triangle connectivity and $k$-truss, aiming to return the maximal $k$-truss subgraph. SM \cite{DBLP:conf/webi/LuoWP06} finds the community with maximum subgraph modularity based on subgraph degree. LM \cite{clauset2005finding} explores the community with maximum local modularity based on boundaries. \textit{HK-Relax} \cite{DBLP:conf/kdd/KlosterG14,hk1,hk2} detects the community by the random walk diffusion. \textit{PPRCS} is our basic algorithm obtained by improving \textit{PPR\&NIBBLE} \cite{DBLP:conf/focs/AndersenCL06,ppr1} (Section \ref{subsec:The CSPPR Algorithm}). \textit{SCCS} is our advanced and efficient algorithm (Section \ref{sec:LCS}). Note that \textit{CSM} and \textit{TCP} focus on capturing dense communities using $k$-core and $k$-truss, respectively, \textit{SM} and \textit{LM} find communities through defined modified modularity, and \textit{HK-Relax}, \textit{PPRCS} and \textit{SCCS} utilize \textit{conductance} as the objective function for identifying resultant community.

\begin{table}[t!] 
\centering
\caption{\small Dataset statistics. $d$ is the diameter of the network.} \vspace{-0.3cm}
\scalebox{1}{
\begin{tabular}{c|ccccc}
\toprule
Dataset & $Type$ & $n$ & $m$  &$d$\\
\midrule
Amazon & Purchasing & 334,863  &925,872  & 44\\ 

DBLP & Authorship & 317,080  &1,049,866  & 21\\

Youtube &Social  &1,134,890  & 2,987,624 &20\\ 	 

LiveJ & Social  &3,997,962 & 34,681,189 & 17 \\

Orkut & Social &3,072,441 & 117,185,083 & 9\\

Friendster & Social & 65,608,366 & 1,806,067,135 & 32\\
\bottomrule			\end{tabular}}\vspace{-0.3cm}
\label{tab:data}
\end{table}

\begin{table*}[t!]
\centering
\caption{Running time (second) of various methods. The best running times are marked in \textbf{bold}. $--$ means that the method cannot obtain results within five hours. AVG.RANK is the average rank of each method across testing datasets.} 
\scalebox{1}{
\begin{tabular}{c|ccccccc}
\hline
\textbf{Model} &Aamazon &DBLP &Youtube &LiveJ &Orkut &Friendster &AVG.RANK\\
\hline
\textit{CSM} &5.690 &6.932 &16.648 &150.534 &838.437 &$--$&6 \\	
\textit{TCP} &34.535 &86.653  &202.892 &4,112.539 & $--$&$--$ & 7\\
\hline
\textit{SM} &\textbf{0.001} &\textbf{0.002} &12.930 &1.524 & 1,130.617 &43.820&4\\
\textit{LM} &\textbf{0.001} &\textbf{0.002} &17.872 &\textbf{0.125}& 3,288.389 &5.927&3\\
\textit{HK-Relax} &0.055 &2.304 &49.684 &87.447& 562.709 &$--$ &5\\
\hline
\textit{PPRCS} &0.080 &0.040  &\textbf{0.387} &0.517 & \textbf{0.877} &\textbf{0.629} &\textbf{1}\\
\textit{SCCS} &0.004 &0.010 &7.982 &2.645 & 84.672 &11.418& 2 \\	
\hline
\end{tabular}}
\label{tab:running}
\end{table*} 

\stitle{Effectiveness Metrics.}  We adopt \textit{conductance} (Definition \ref{def:co}) and \textit{F1-score} to evaluate the quality of communities. \textit{Conductance} measures the internal cohesion and external sparsity of a community \cite{DBLP:conf/kdd/KlosterG14,DBLP:conf/focs/AndersenCL06,DBLP:journals/pvldb/WuJLZ15}. Lower \textit{conductance} indicates higher cohesion within the community and fewer connections outside.   \textit{F1-score} is a metric that considers both \textit{precision} and \textit{recall}. \textit{F1-score} is defined as:
 \begin{equation}
 F1-score = \frac{2\times Precision\times Recall}{Precision+Recall}
 \end{equation}
 where \textit{precision} and \textit{recall} are computed as follows:
\begin{flalign}
Precision &= \frac{True~Positives}{True~Positives+False~Positives} 
\\
Recall &= \frac{True~Positives}{True~Positives+ False~Negatives}
\end{flalign}

We compare the discovered community $S$ with the ground-truth community to evaluate algorithm performance. Higher similarity between $S$ and ground truth leads to higher \textit{precision}, \textit{recall}, and \textit{F1-score}. \textit{F1-score} ranges from 0 to 1, with values closer to 1 indicating better performance. Besides, we also measure the community size.

\stitle{Parameters.} Consistent with previous research \cite{DBLP:conf/focs/AndersenCL06}, we set $\alpha = 0.15$ and $r_{max} = 1/n$ as the default parameters for our basic algorithm \textit{PPRCS}. Similarly, following \cite{DBLP:journals/tkdd/HeSBH19}, we default the minimum and maximum vertex counts for our advanced algorithm \textit{SCCS} to $l=300$ and $h=5000$, respectively. Besides, we randomly select 50 query vertices from 50 ground truth communities to ensure reliability and report average running time and quality.

\subsection{Empirical Results on Real-world Graphs} 
\stitle{Exp-1: Running time of different methods.} Table \ref{tab:running} shows the running time of different algorithms on various datasets. We observe the following conclusions: (1) Our algorithms (i.e., \emph{PPRCS} and \emph{SCCS}) are on average faster than other methods. In particular, \emph{PPRCS} is the champion and \textit{SCCS} is the runner-up, while \emph{SCCS} has better \textit{F1-score} than \emph{PPRCS} (Table \ref{table:metric}). On the Friendster dataset with 1.8 billion edges, our \emph{PPRCS} (resp., \emph{SCCS}) takes only about 0.6 (resp., 11) seconds,
while \emph{CSM}, \emph{TCP}, and \emph{HK-Relax} fail to produce results within five hours. (2) \textit{CSM} and \textit{TCP} have the worst running time. This is because \textit{CSM} is sensitive to the choice of query vertices. Specifically, the higher the degree of the query vertex, the more iterations of \textit{CSM} and the slower the convergence speed. On the other hand, \textit{TCP} needs to count the number of triangles that each edge participates in, resulting in poor efficiency for massive graphs. (3) \textit{SM} and \textit{LM} have faster running time on small datasets (i.e., Aamazon and DBLP) because they optimize community structures by only focusing on changes in node-edge connections within the neighborhood. However, \textit{SCCS} requires additional time in the initial stage to find the largest clique as the initial community and in the verification stage to check the connectivity of the community. Fortunately, on large graphs (such as Orkut), \textit{SCCS} is much more efficient than \textit{SM} and \textit{LM} due to its subgraph sampling strategy. Summing up, our algorithms have a more stable running time and are less influenced by graph size  when contrasted with baselines.

\stitle{Exp-2: Community size of different methods.} Previous studies \cite{DBLP:conf/www/LeskovecLDM08,DBLP:journals/tbd/WangYBH24} have indicated that real-world communities tend to be relatively small. Building upon this observation, we conduct a statistical analysis of community sizes (Table \ref{tab:range}). The results show that, except Orkut, at least 90\% of the community sizes falls on [3,100]. To verify whether the results of different algorithms match the real-world community sizes, Table \ref{tab:size} displays the community sizes captured by various algorithms on different datasets. We have the following observations: (1) Communities identified by \textit{CSM} and \textit{TCP} tend to be large in size due to the fact that they are only constrained by internal community connections. Specifically, \textit{CSM} stops expanding community $S$ when the maximum-minimum degree of the vertices in $S$ no longer increases, while \textit{TCP} stops when expanding $S$ no longer satisfies triangular connectivity. (2) The community sizes of other algorithms are smaller than those of \textit{CSM} and \textit{TCP} because they are constrained by both internal and external connections of the community, resulting in more stringent conditions for community expansion. (3) \textit{PPRCS} and \textit{HK-Relax} perform the graph diffusion, which tend to produce large communities that often do not align with the characteristics of real-world communities. In contrast, \textit{SM}, \textit{LM} and \textit{SCCS} are all more suitable for identifying small-scale communities by optimizing their objective functions.

\begin{table}[t!] 
\centering
\caption{\small Size distribution of ground-truth community.} \vspace{-0.3cm}
\fontsize{8}{10}\selectfont
\scalebox{1}{
\begin{tabular}{c|cccccc}
\toprule
Dataset & $[3, 50)$ & $[50, 100)$ & $[100, 150)$  &$[150, 200)$ &  $200\leq $\\
\midrule
Amazon & 4,836 & 132  & 25 & 3 & 4 \\ 
DBLP & 4,937 & 9  & 5  & 4 & 45 \\
Youtube & 3,118  & 129 & 40 & 19 & 49 \\
LiveJ & 4,424  & 430  & 69 & 27 & 50  \\ 	 
Orkut & 630 & 1,579 & 807 & 449 & 1,535 \\
Friendster & 3,430 & 1,114 & 257 & 77 & 122 \\
\bottomrule			\end{tabular}}\vspace{-0.3cm}
\label{tab:range}
\end{table}

\begin{table}[t!]
\centering
\caption{\small Community sizes captured by different algorithms on different datasets. $--$ means that the method cannot obtain results within five hours.} \vspace{-0.3cm}
\fontsize{7}{10}\selectfont
\scalebox{0.95}{
\begin{tabular}{c|cccccc}
\toprule
\multicolumn{1}{c|}	{Methods} & Amazon& DBLP& Youtube& LiveJ& Orkut & Friendster\\
\midrule		
\textit{CSM}& 58,426 & 46,141 & 79,160 & 728,885 & 1,189,342 & $--$\\
\textit{TCP}& 10 & 16,057 & 105,360 & 541,491 & 2,622,574 &$--$\\	
\textit{SM}& 10 & 10 & 47 & 30 & 136&46 \\
\textit{LM}  & 9 & 8 &  35 & 20 & 88 &25\\
\textit{HK-Relax} & 20 & 3,980 & 370,685 & 397,628 & 876,618& $--$\\
\textit{PPRCS} & 36 & 51 & 6,579 & 468 & 3,622 & 78,044\\
\textit{SCCS} & 11 & 10 & 254 & 27 &  122&69\\
\bottomrule	
\end{tabular}} \vspace{-0.3cm}
\label{tab:size}
\end{table}

\stitle{Exp-3: Community quality of different methods.} Table \ref{table:metric} reports the \textit{F1-score} and \textit{conductance} of different algorithms on various datasets. We have the following insights: (1) Our \textit{SCCS} exhibits the best \textit{F1-score}, which means that \textit{SCCS} can more accurately identify the ground-truth communities compared to other competitors. This superiority stems from \textit{SCCS}'s initial focus on capturing maximal cliques, ensuring tight connectivity within communities. Subsequently, \textit{SCCS} takes into account both the closeness within the community and the sparsity between communities in the expansion and verification stages, further ensuring the separation between different communities. Additionally, we allow up to $count$ vertices to be added at a time during the expansion stage, thus preventing the community search from prematurely falling into local optima (Section \ref{subsec:Community Extraction}). 
(2) \textit{CSM} and \textit{TCP} exhibit the worst  \textit{conductance} due to their focus solely on the tightness within communities, neglecting the sparsity between communities. In contrast to \textit{CSM} and \textit{TCP}, the remaining algorithms consider connections within and outside the community simultaneously. For example, \textit{SM} and \textit{LM} are based on the graph structure and greedily add one vertex at a time that makes the objective function optimal, while \textit{SCCS} relaxes this constraint by allowing the consideration of adding multiple vertices to communities at once, resulting in lower \textit{conductance}. (3) Both \textit{HK-Relax} and \textit{PPRCS} rely on random walks to produce communities with better \textit{conductance}, but \textit{SCCS} achieves smaller communities (Table \ref{tab:size}) and superior \textit{F1-score}. So, \textit{SCCS} outperforms other algorithms in capturing the qualities of real-world communities. 

\begin{table*}[t!]
\caption{\textit{F1-score} and \textit{Conductance} (\textit{Con}) on different datasets. The best result in each metric are marked in \textbf{bold}. $--$ means that the method cannot obtain results within five hours. AVG.RANK is the average rank of each methods across the testing datasets.} 
\centering
\scalebox{0.9}{
\begin{tabular}
{c|cccccccccccccc}
\hline
\multirow{2}{*}{\textbf{Model}} & 
\multicolumn{2}{c}{\textbf{Amazon}} & \multicolumn{2}{c}{\textbf{DBLP}} & \multicolumn{2}{c}{\textbf{Youtube}} &
\multicolumn{2}{c}{\textbf{LiveJ}} &
\multicolumn{2}{c}{\textbf{Orkut}} &
\multicolumn{2}{c}
{\textbf{Friendster}} &
\multicolumn{2}{c}{\textbf{AVG.RANK}}
\\
&\textit{F1-score} & \textit{Con} & \textit{F1-score} & \textit{Con} & \textit{F1-score} & \textit{Con} & \textit{F1-score} & \textit{Con} & \textit{F1-score} & \textit{Con} & \textit{F1-score} & \textit{Con} & \textit{F1-score} & \textit{Con}\\
\hline	
\textit{CSM} & 0.607 & 0.262 & 0.417 & 0.412 & 0.060 & 0.621 & 0.169 & 0.529 & 0.000 & 0.480 & $--$ & $--$ & 7 & 7  \\
\textit{TCP} & 0.881 & 0.114 & 0.598 & 0.298 & 0.170 & 0.596 & 0.618 & 0.272 & 0.090 & 0.940 & $--$ & $--$ & 4 & 6 \\
\hline	
\textit{SM} & 0.879 & 0.085 & 0.650 & 0.225 & 0.204 & 0.524 & 0.651 & 0.192 & 0.354 & 0.536 & 0.301 & 0.685 & 2 &4\\
\textit{LM} & 0.865 & 0.112 & 0.610 & 0.287 & 0.202 & 0.552 & 0.652 & 0.239 &0.215 & 0.656 & 0.322 & 0.715 & 3 &5\\
\textit{HK-Relax} & 0.875 & 0.046 & 0.532 & \textbf{0.132} & 0.020 & \textbf{0.175} & 0.521 & \textbf{0.107} & 0.045& \textbf{0.137} & $--$ & $--$ & 6 & \textbf{1} \\
\hline
\textit{PPRCS} & 0.850 & \textbf{0.045} & 0.530 & 0.134 & 0.125 & 0.412 & 0.542 & 0.146 & 0.382 & 0.570 & 0.202 & \textbf{0.629} & 5 & 2\\
\textit{SCCS} & \textbf{0.905} & 0.063 & \textbf{0.696} & 0.187 & \textbf{0.241} & 0.550 & \textbf{0.695} & 0.188 & \textbf{0.523}& 0.598 & \textbf{0.492} & 0.735 & \textbf{1} & 3 \\
\hline	
\end{tabular}}
\label{table:metric}
\end{table*}



\stitle{Exp-4: Running time and quality of \textit{SCCS} with varying parameters.} We compare the running time of \textit{SCCS} on different query sets. We pre-classify the vertices in the real-world communities into three types based on degree: dense (degree greater than the average degree), random, and sparse (degree less than the average degree). As shown in Figure \ref{fig:query}, \textit{SCCS} exhibits fastest running time in sparse query sets. This can be attributed to the fact that vertices in sparse sets require to consider fewer neighbors at each stage, allowing \textit{SCCS} to converge faster. 
In contrast, \textit{SCCS} has the slowest running time in dense query sets since vertices in dense query sets have many edges and more neighbors need to be considered at each stage. Random query sets refer to a random mix of sparse and dense query sets, so the corresponding running time falls between these two extremes.

\begin{figure}[t!]
\centering	
\subfigure[\textit{Runtime}]{
\includegraphics[width=0.22\textwidth]{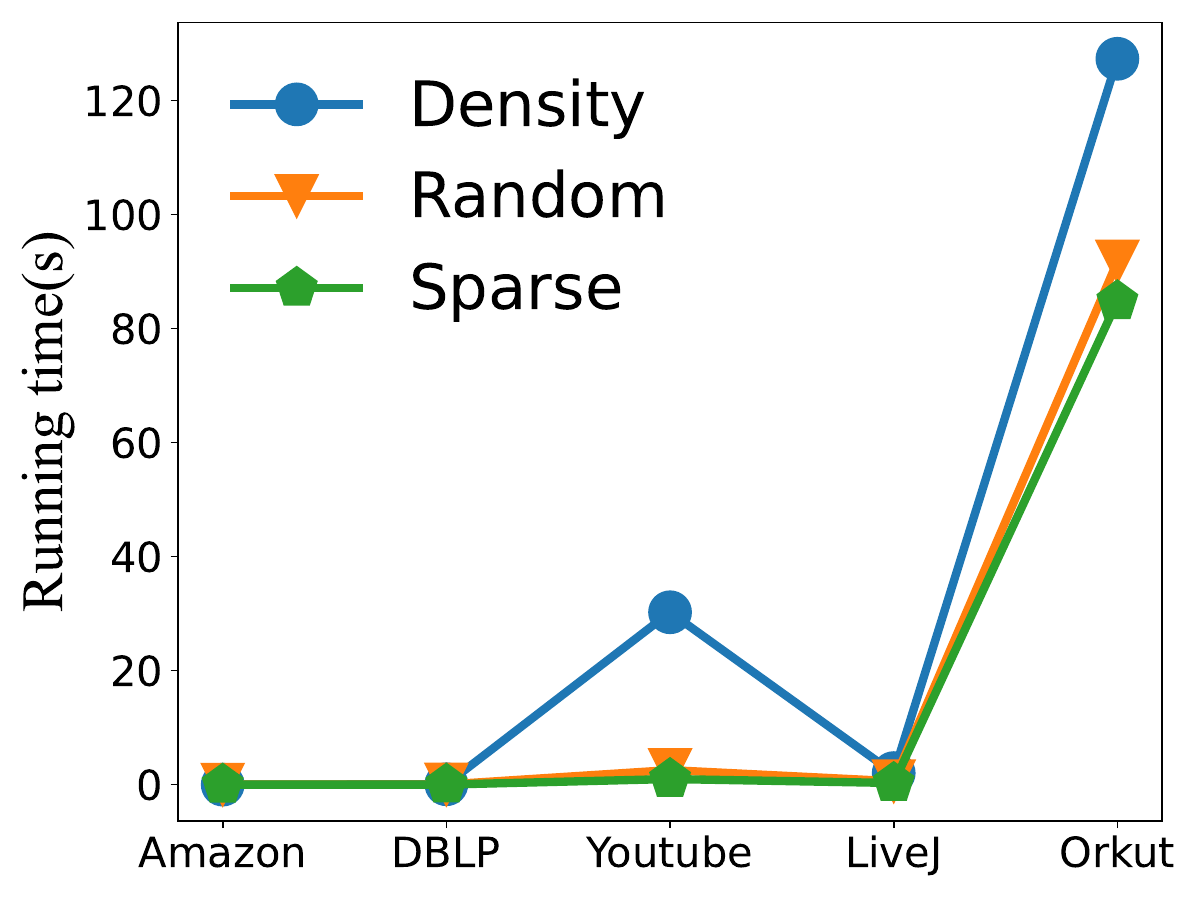}
\label{fig:query(a)}}
\subfigure[\textit{Runtime}]{
\includegraphics[width=0.22\textwidth]{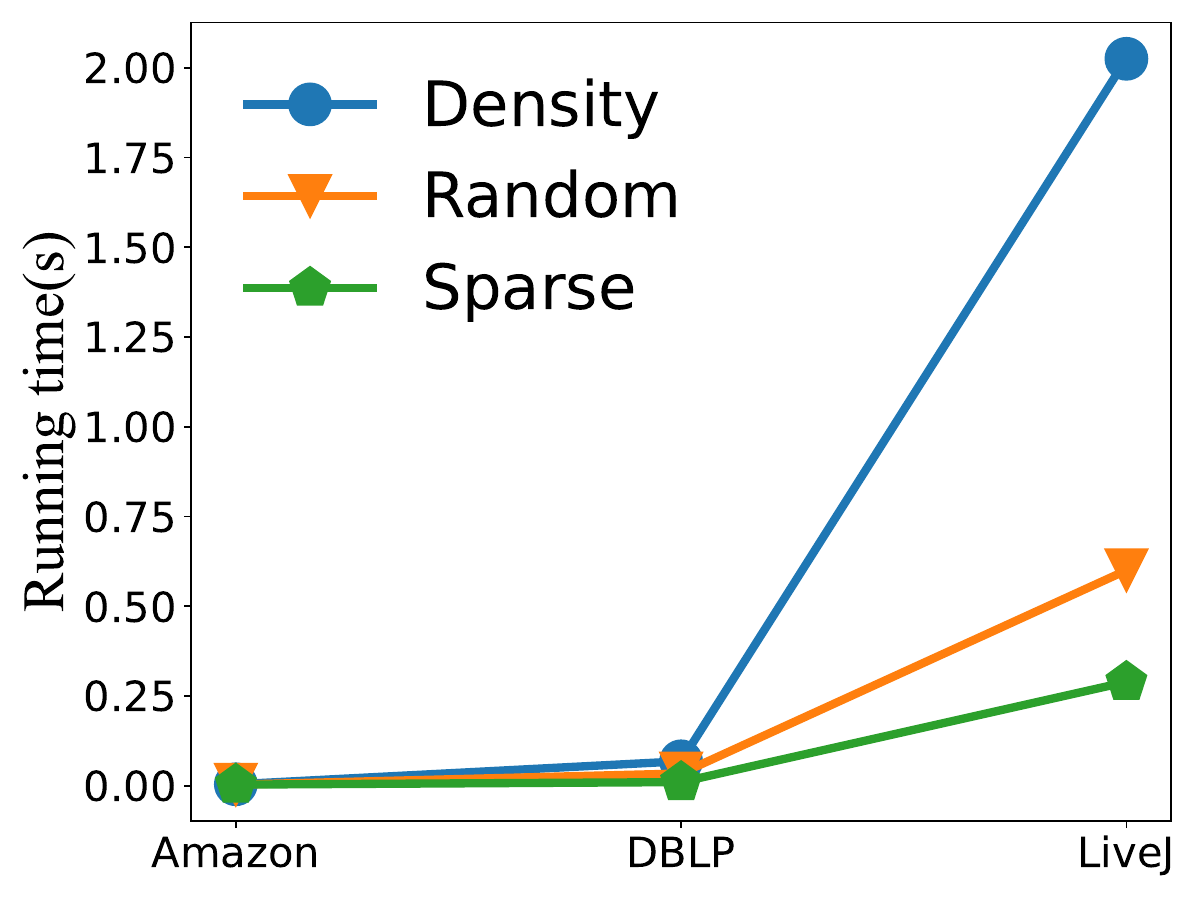}
\label{fig:query(b)}
}\vspace*{-0.3cm}
\caption{Running time of \emph{SCCS} on different query sets. (b) is an enlarged part of (a).}
\label{fig:query} \vspace{-0.5cm}
\end{figure}


On top of that, we also investigate the effect of different parameters $count$ and $dp$ on the performance of \emph{SCCS}. Specifically, we first evaluate the \textit{conductance} and \textit{F1-score} of \textit{SCCS} on Amazon and DBLP for different $count$. Similar trends can be found in other datasets. As shown in Figure \ref{fig:count}, with the  parameter $count$ increases, the \textit{conductance} linearly decreases, while the \textit{F1-score} initially increases and then decreases. The reason for this trend is that as $count$ increases, communities avoid prematurely converging to local optima, leading to a slight decrease in \textit{conductance} but an increase in community size. As the predicted community size increases, more true community members may be included, thereby improving \textit{recall}, which is the proportion of true community members successfully detected. However, with the increase in community size, there is a possibility of misclassifying some vertices that do not belong to the true community as community members, leading to a decrease in \textit{precision}. Therefore, considering that real communities typically have smaller sizes, larger predicted communities often exhibit higher \textit{recall} but lower \textit{precision}, leading to a decrease in \textit{F1-score}. So, we recommend users set $count = 2$  as the default for better performance. 

On the other hand,  we study the subgraph coverage, vertex sampling rate, and sampling time of target communities (i.e., $S$) at sampling depth $dp\in [1,2,3,4,5]$ on all datasets with ground-truth communities (i.e., $S_{Truth}$). $Vertex~sampling~rate = n_{S}/n$ and $Subgraph~coverage= |S_{Truth} \cap S|/|S_{Truth}|$, where a higher coverage indicates that the sampled subgraph is more representative. As shown in Table \ref{table:dp}, subgraph coverage, vertex sampling rate, and sampling time initially increase linearly with $dp$ and then level off. The reason for the latter is that the number of sampled vertices is limited to not exceed $h$, which also results in low vertex sampling rates and sampling times. We observe on YouTube that when $dp$ increases from 2 to 3, the coverage improves by 6\%, but when $dp$ increases from 3 to 4, the coverage only improves by 0.2\%. Thus, we recommend users set $dp=3$  as the default for better performance.

\begin{table*}[t!]
\caption{Subgraph coverage (Coverage for short), vertex sampling rate (Rate for short), and sampling time (second) of sampled subgraphs at different $dp$. $--$ means that the result is less than 0.0001. $Coverage = |S_{Truth} \cap S| / S_{Truth}$ and $Rate = n_{S} / n$.} 
\centering
\fontsize{7}{7}\selectfont 
\scalebox{0.9}{
\begin{tabular}
{c|cccccccccccc}
\hline
\multirow{2}{*}{\textbf{Depth}} & 
\multicolumn{2}{c}{\textbf{Amazon}} & \multicolumn{2}{c}{\textbf{DBLP}} & \multicolumn{2}{c}{\textbf{Youtube}} & 
\multicolumn{2}{c}{\textbf{LiveJ}} &
\multicolumn{2}{c}{\textbf{Orkut}}&
\multicolumn{2}{c}{\textbf{Friendster}}\\
& Coverage/Rate & Time & Coverage/Rate & Time & Coverage/Rate & Time & Coverage/Rate & Time & Coverage/Rate & Time & Coverage/Rate & Time\\
\hline	
$dp = 1$ & 0.997/0.0008 & 0.001 & 0.964/0.0009 & 0.004 & 0.783/0.0002 & 0.049 & 0.950/$--$ & 0.017 & 0.285/$--$ & 0.105 &0.579/$--$ &0.088 \\

$dp = 2$ & 0.997/0.0008 & 0.001 & 0.964/0.0009 & 0.002 & 0.823/0.0005 & 0.062 & 0.973/$--$ & 0.017 & 0.640/0.0010 & 0.865& 0.812/$--$ &1.097\\

$dp = 3$ & 0.997/0.0008 & 0.001 & 0.978/0.0014 & 0.003 & 0.883/0.0029 & 0.211 & 0.991/0.0002 & 0.055 & 0.650/0.0015 & 0.879 &0.817/$--$ &1.667\\	

$dp = 4$ & 0.997/0.0009 & 0.001 & 0.989/0.0070 & 0.018 & 0.885/0.0041 & 0.301 & 0.992/0.0009 & 0.209 & 0.650/0.0016 & 0.873 & 0.817/$--$& 1.722\\

$dp = 5$ & 0.997/0.0013 & 0.001 & 0.989/0.0128 & 0.034 & 0.886/0.0044 & 0.317 & 0.992/0.0011 & 0.281 & 0.650/0.0016 & 0.878 & 0.817/$--$ & 1.740\\
\hline	
\end{tabular}}
\label{table:dp}
\end{table*}

\begin{figure}[t!]
\centering
\subfigure[\textit{Conductance}]{	\includegraphics[width=0.22\textwidth]{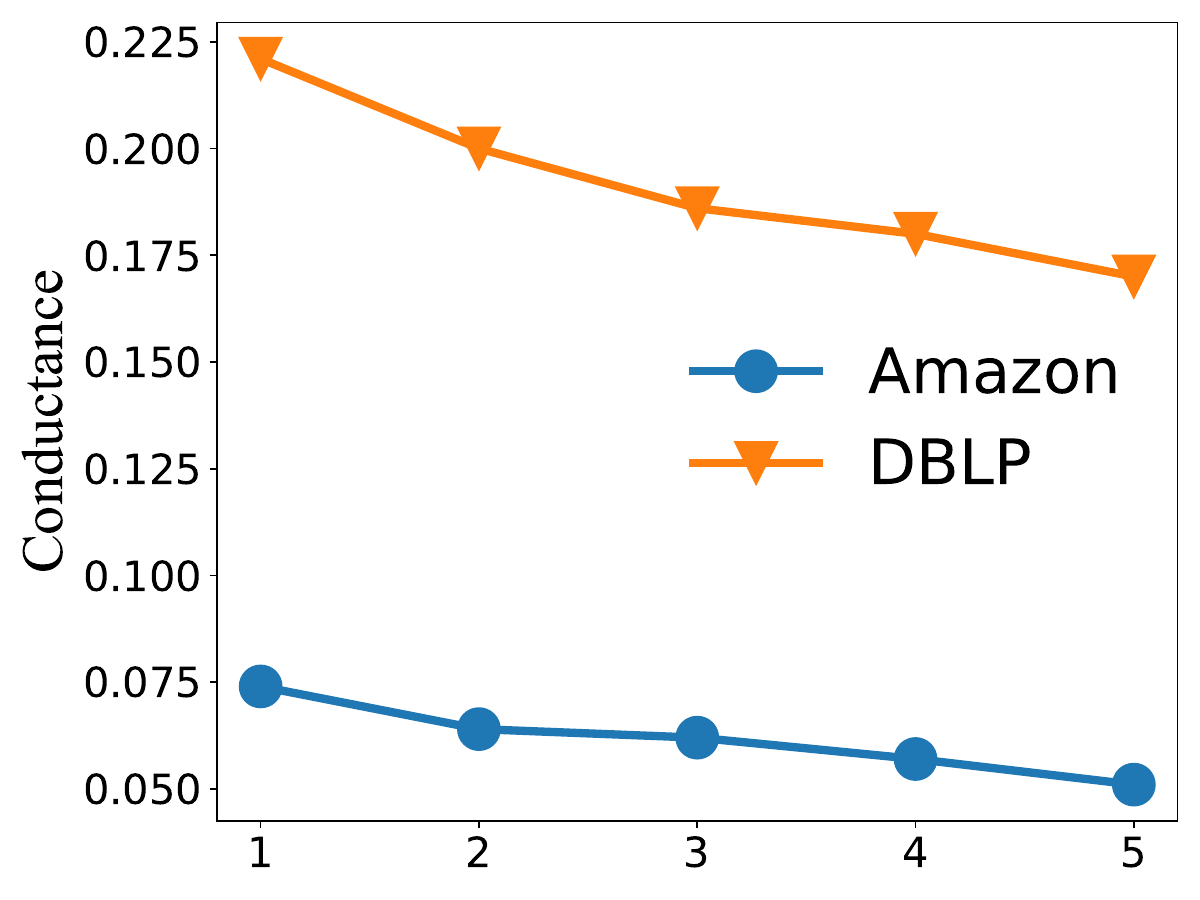}
\label{fig:count(a)}}
\subfigure[\textit{F1-score}]{	\includegraphics[width=0.22\textwidth]{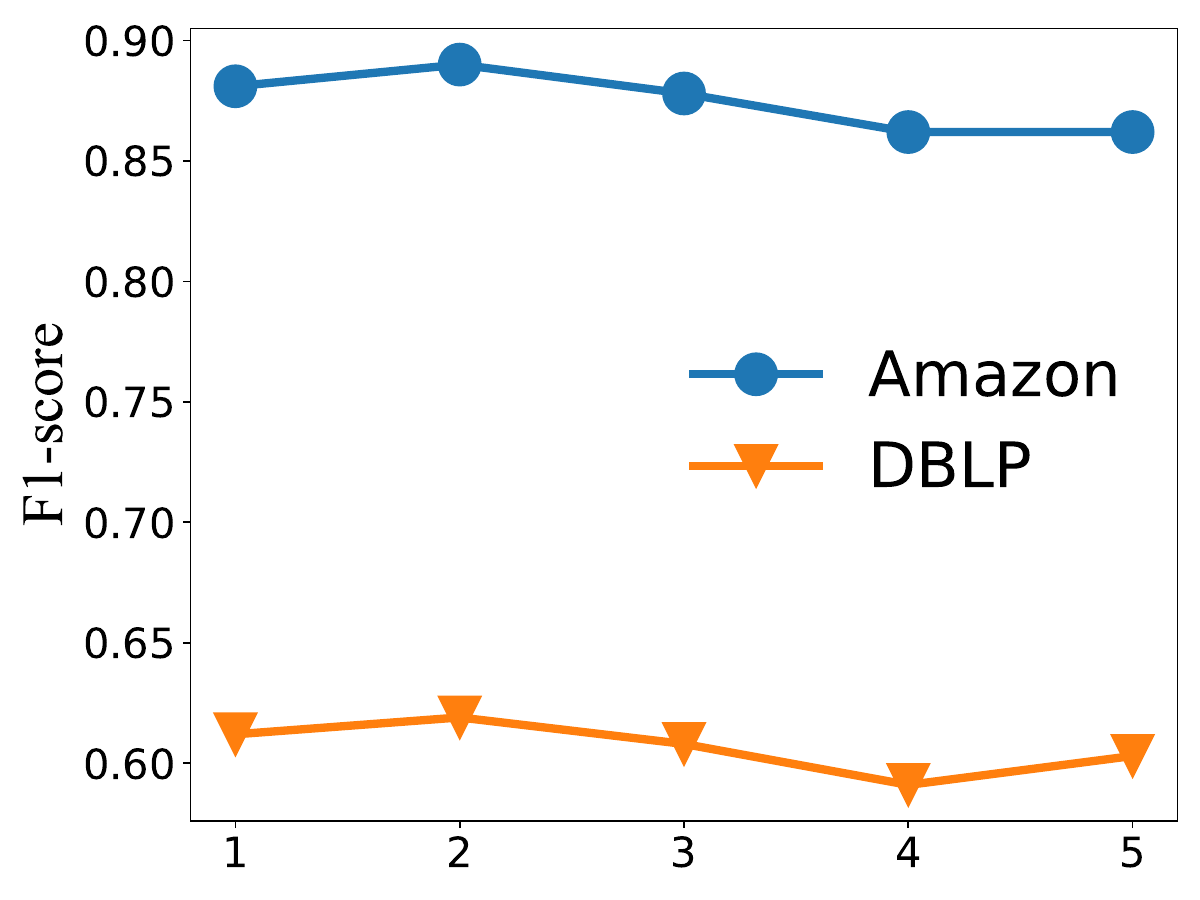}
\label{fig:count(b)}}\vspace*{-0.3cm}
\caption{The performance of \emph{SCCS} with varying the parameter $count$.}
\label{fig:count} \vspace{-0.5cm}
\end{figure}

\stitle{Exp-5: Ablation study  of the maximal clique.} In theory, we need to designate each maximal clique in $initialCom$ as an initial community, then apply expansion and contraction paradigms, and finally select the community with the minimum \textit{conductance} (or maximum \textit{F1-score}) as the target community. The advantage of this approach is that it allows us to consider all possible communities containing the query vertex $q$. However, experimental findings suggest that it incurs significant time overhead. Specifically, Figure \ref{fig:clique} compares the \textit{F1-score} and running time of \textit{SCCS} under two initial community schemes: including all maximal cliques in $initialCom$ (\textit{SCCS1}) and selecting only one maximal clique from $initialCom$ (\textit{SCCS2}). Compared to the latter, the former only slightly improves the F1-score, while the running time increases significantly. Therefore, to balance \textit{F1-score} and computational cost, we choose to return only one maximal clique as the initial community, which is consistent with our theoretical analysis in Section \ref{sec:LCS}.

\begin{figure}[t!]
\centering
\subfigure[\textit{F1-score}]{
\includegraphics[width=0.22\textwidth]{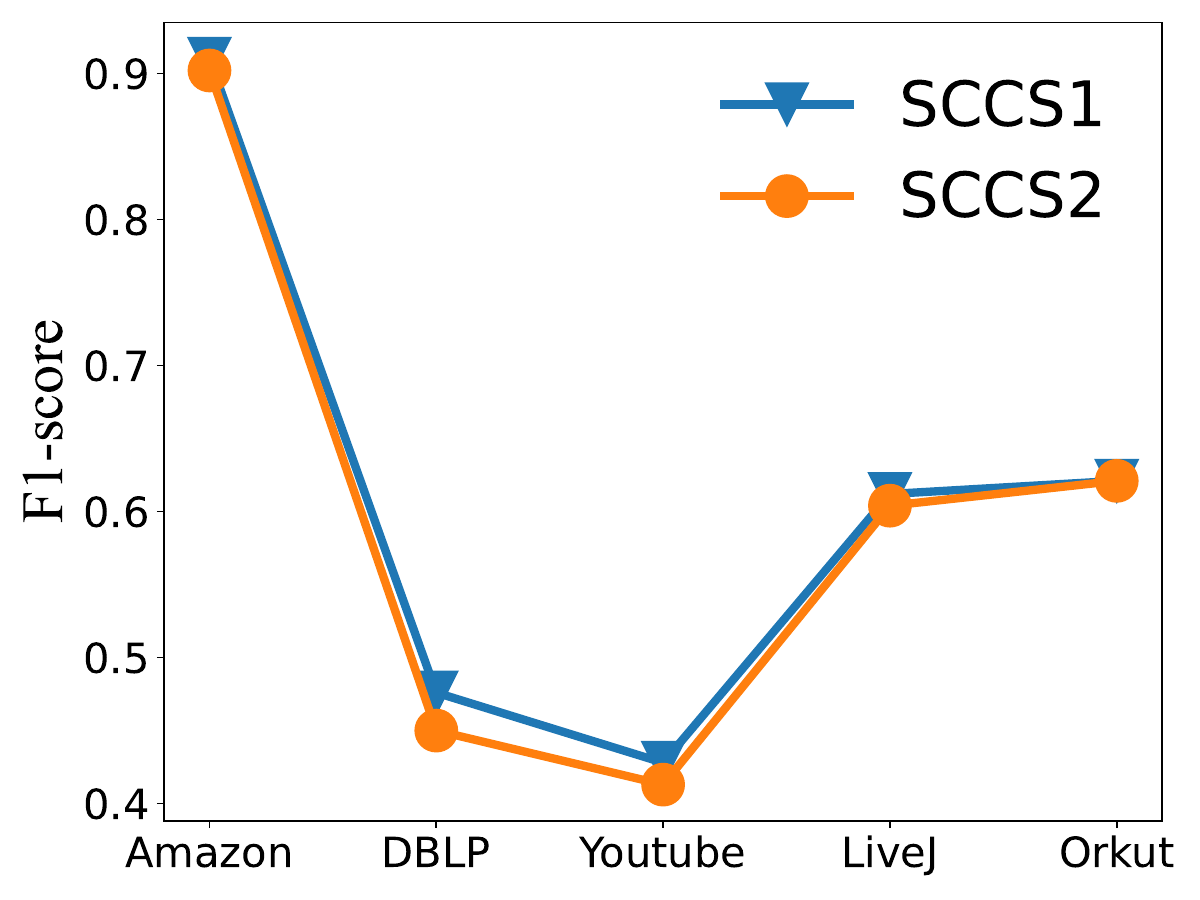}
\label{fig:clique(a)}}
\subfigure[\textit{Runtime}]{
\includegraphics[width=0.22\textwidth]{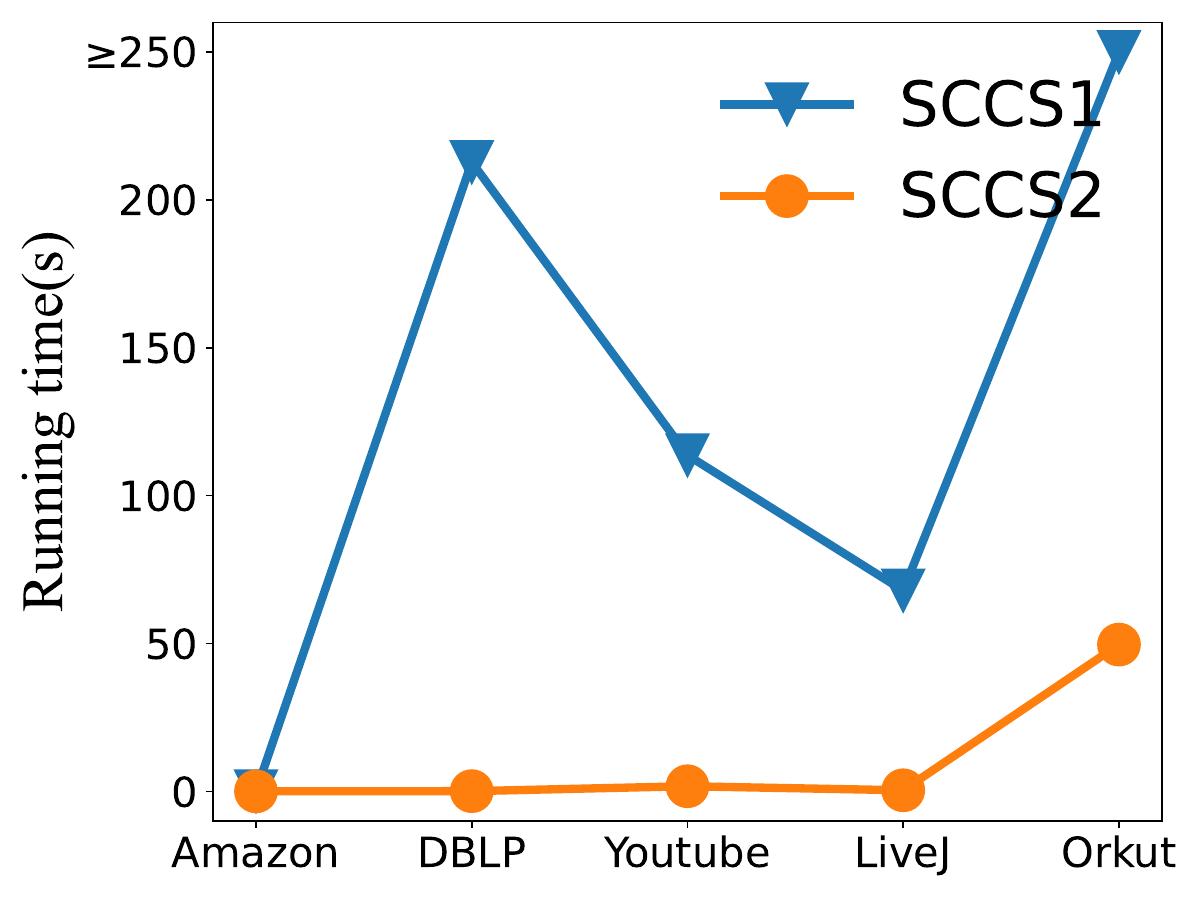}
\label{fig:clique(b)}
}\vspace*{-0.3cm}
\caption{Ablation study  of the maximal clique.}
\label{fig:clique} \vspace{-0.5cm}
\end{figure}

\begin{figure*}[t!]
\centering
\subfigure[\textit{PPRCS}]{
\includegraphics[width=0.22\textwidth]{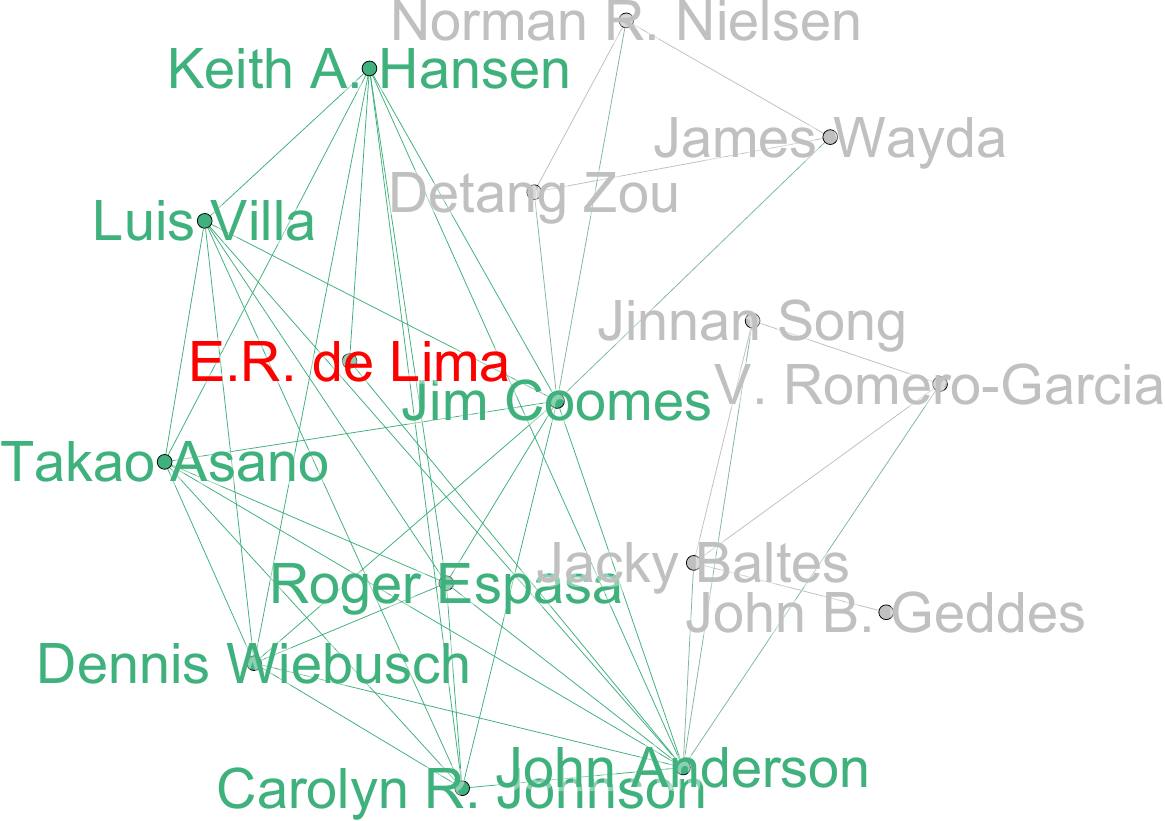}
}
\subfigure[\textit{SCCS}]{
\includegraphics[width=0.22\textwidth]{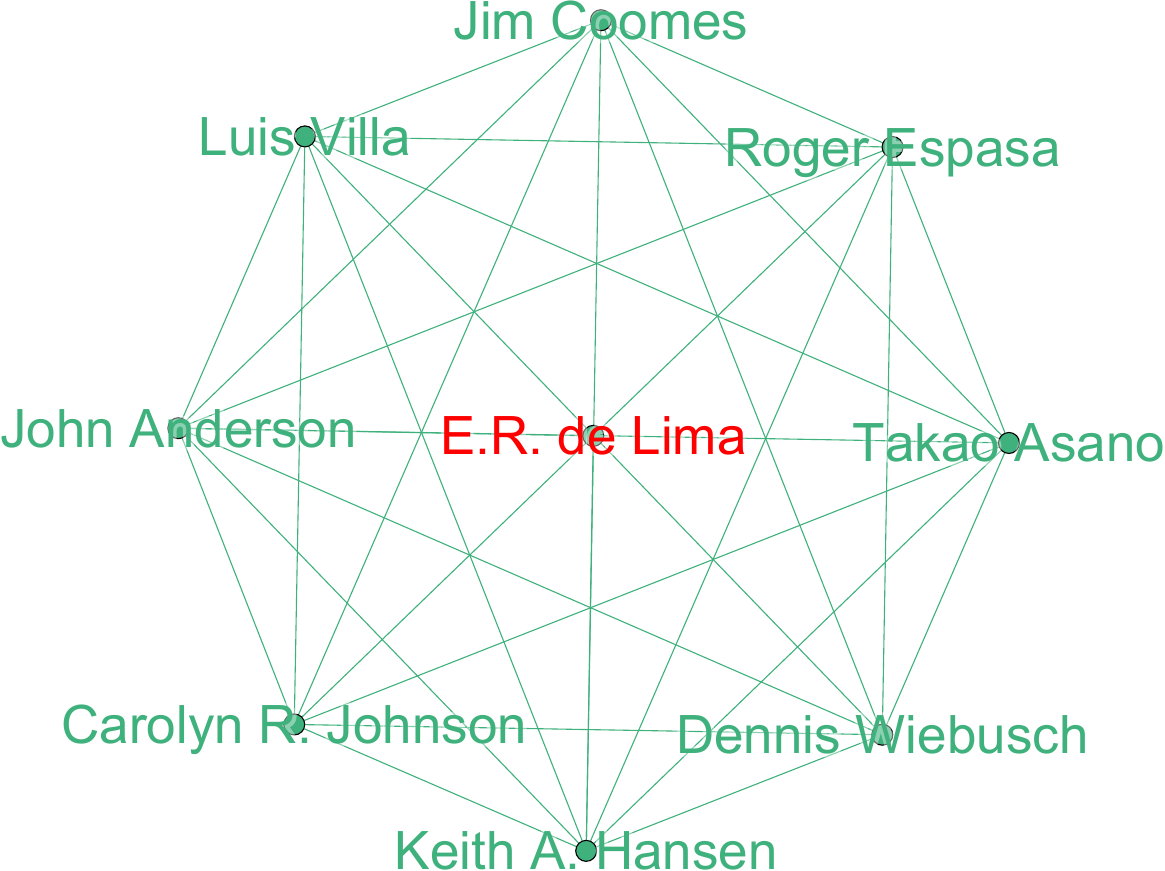}
}
\subfigure[\textit{PPRCS}]{
\includegraphics[width=0.22\textwidth]{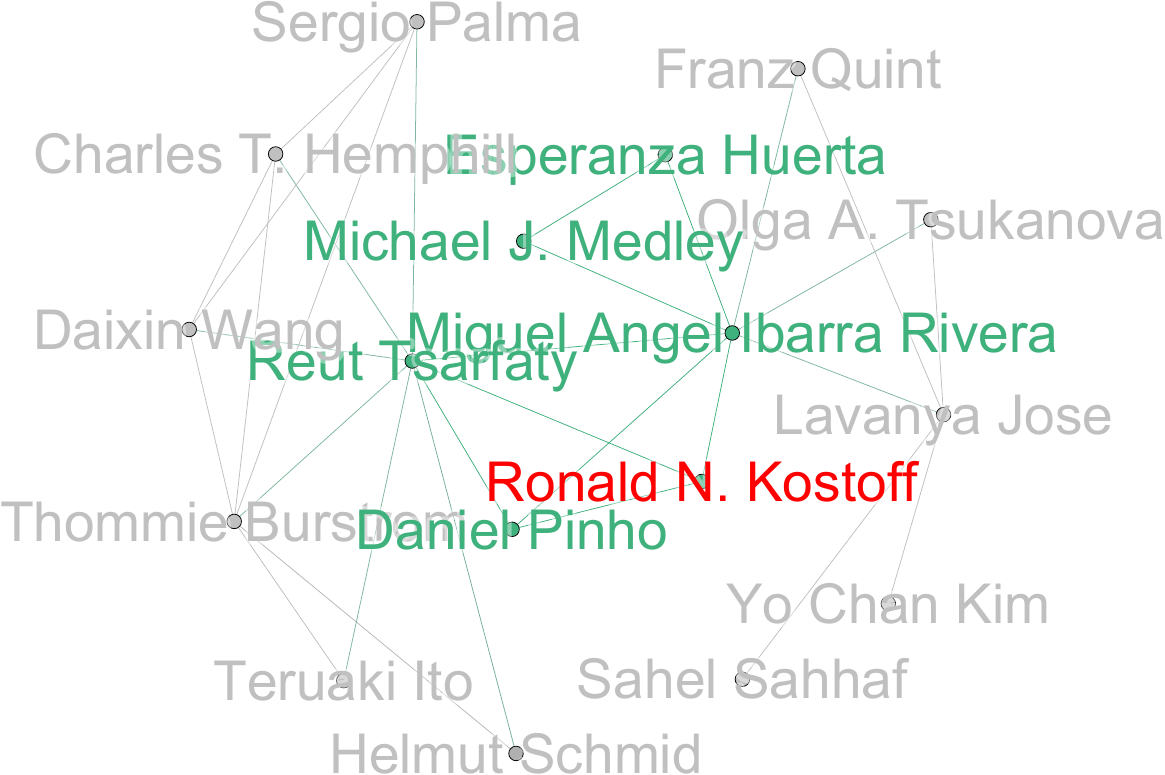}
}
\subfigure[\textit{SCCS}]{
\includegraphics[width=0.22\textwidth]{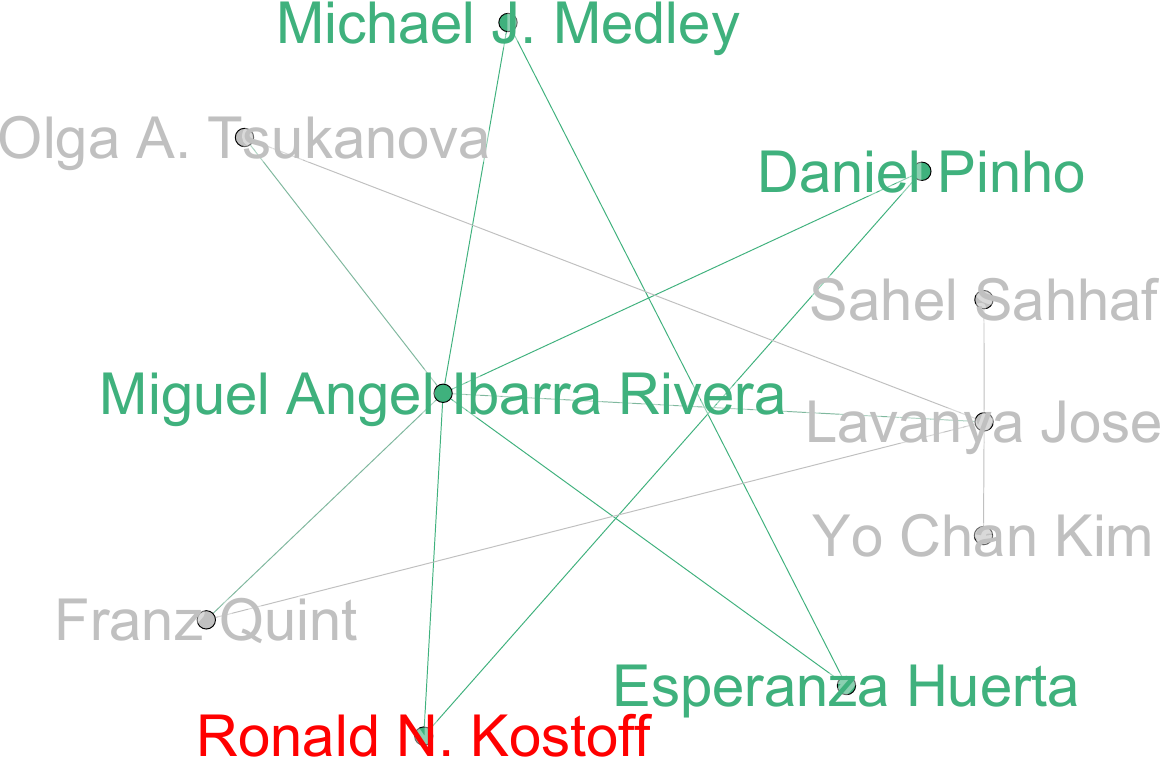}
}
\caption{Case studies on DBLP.}
\label{fig:case} 
\end{figure*}

\stitle{Exp-6: Case studies on DBLP.} For the case study, we conduct queries using different vertices based on \textit{PPRCS} and \textit{SCCS} on DBLP. Figures \ref{fig:case} (a-b) show the collaboration networks about the author "$E. R. de Lima$" found by \textit{PPRCS} and \textit{SCCS}, respectively. Obviously, (1) the community size of \textit{PPRCS} is larger than that of \textit{SCCS}, which does not conform to the small-scale properties of real communities. (2) The proportion of real community members in \textit{PPRCS} is 9/16, while the proportion of real community members in \textit{SCCS} is as high as 1, where the green vertices in the network represent members of the real community. From this, the \textit{F1-score} of \textit{PPRCS} and \textit{SCCS} are calculated to be 0.72 and 1 respectively. 
 Figures \ref{fig:case} (c-d) return the collaboration networks containing the author "$Ronald N. Kostoff$" with a similar trend. In particular, we observe that (1) there are three weakly connected partitions in \textit{PPRCS}, indicating that \textit{PPRCS} is more likely to contain three communities. (2) The proportion of real community members in \textit{PPRCS} is 6/17, while the proportion of real community members in \textit{SCCS} is 5/10. Considering both \textit{precision} and \textit{recall}, the \textit{F1-score} of \textit{PPRCS} and \textit{SCCS} are 0.44 and 0.5 respectively. The reason for the above  phenomena is that as the predicted community size increases, the \textit{recall} increases, but the \textit{precision} may decrease, resulting in a decrease in the \textit{F1-score}. Thus, \textit{SCCS} outperforms \textit{PPRCS} in discovering communities that closely match real-world scenarios.

\vspace{-0.2cm}
\subsection{Empirical Results on Synthetic Graphs} 
\label{sec:Scalability}

The synthetic \emph{LFR} \cite{lancichinetti2009detecting} benchmark datasets have three parameters: $n$ (the number of vertices), $ad$ (the average degree), and $\mu$ (community mixing ratio). In particular, A larger $\mu$ implies that the number of edges
crossing different communities increases, resulting in that being more difficult to detect intrinsic communities.

\stitle{Exp-7: Performance on LFRs with varying $n$.}  We fix the average degree $ad=18$ and the mixing ratio $\mu=0.1$, Figure \ref{fig:LFR_n} shows  the performance of \textit{SCCS} and \textit{PPRCS} on LFR networks as $n$ increased from 20,000 to 60,000. We observe the following phenomena: (1) As $n$ increases, the running time and \textit{F1-score} of \textit{SCCS} exhibit only slight fluctuations. (2) Although \textit{PPRCS} is more efficient than \textit{SCCS}, its \textit{F1-score} significantly decreases with the increase in $n$ and is inferior to that of \textit{SCCS}. The reason for the observed phenomena is that \textit{SCCS} trims the original graph into smaller subgraphs through sampling and then conducts a community search on these sampled subgraphs, whereas PPRCS generates large sets based on the graph diffusion. In summary, as the network size $n$ grows, both \textit{PPRCS} and \textit{SCCS} can obtain communities quickly. But, \textit{PPRCS} is less practical compared to \textit{SCCS} because the latter is more effective in discovering communities.

\begin{figure}[t!]
\centering	
\subfigure[\textit{F1-score}]{
\includegraphics[width=0.22\textwidth]{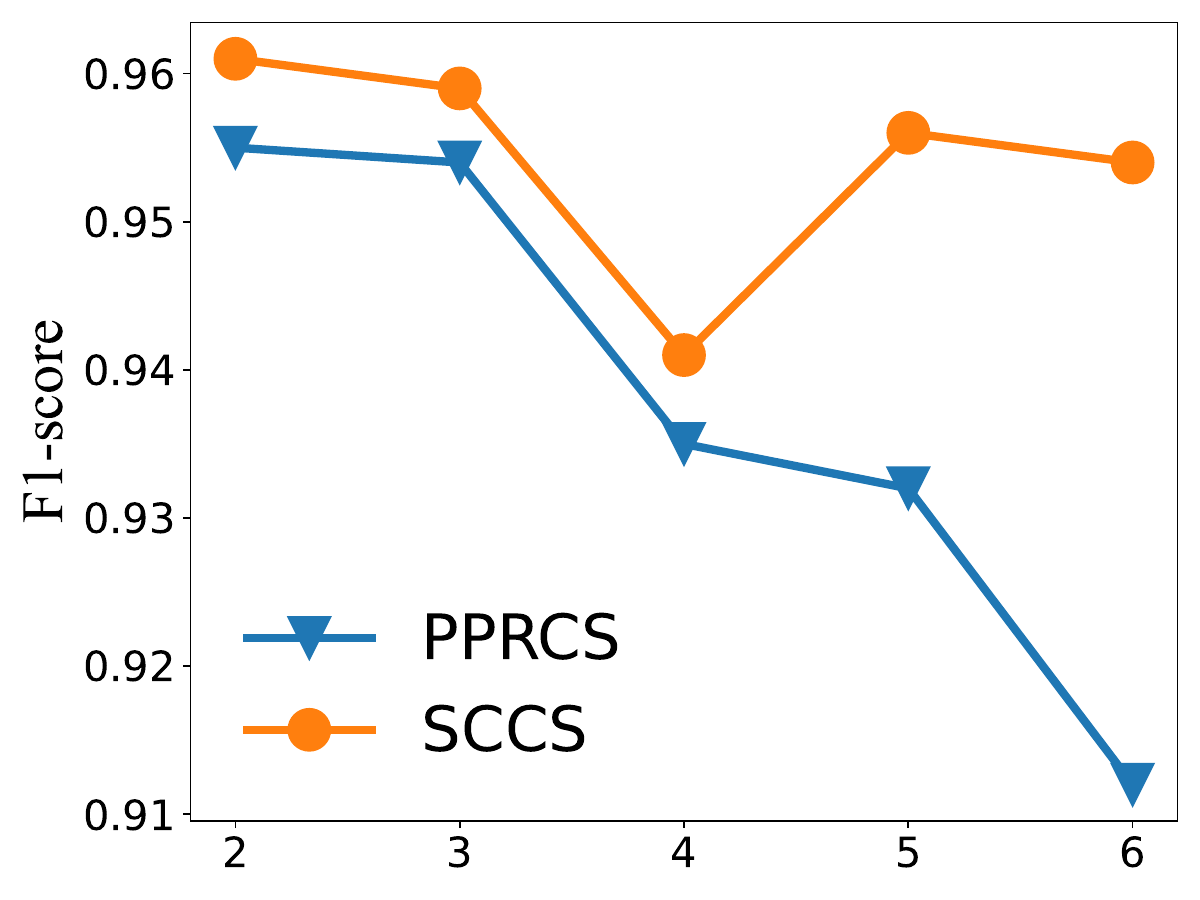}
\label{fig:LFR_n(a)}}
\subfigure[\textit{Runtime}]{
\includegraphics[width=0.22\textwidth]{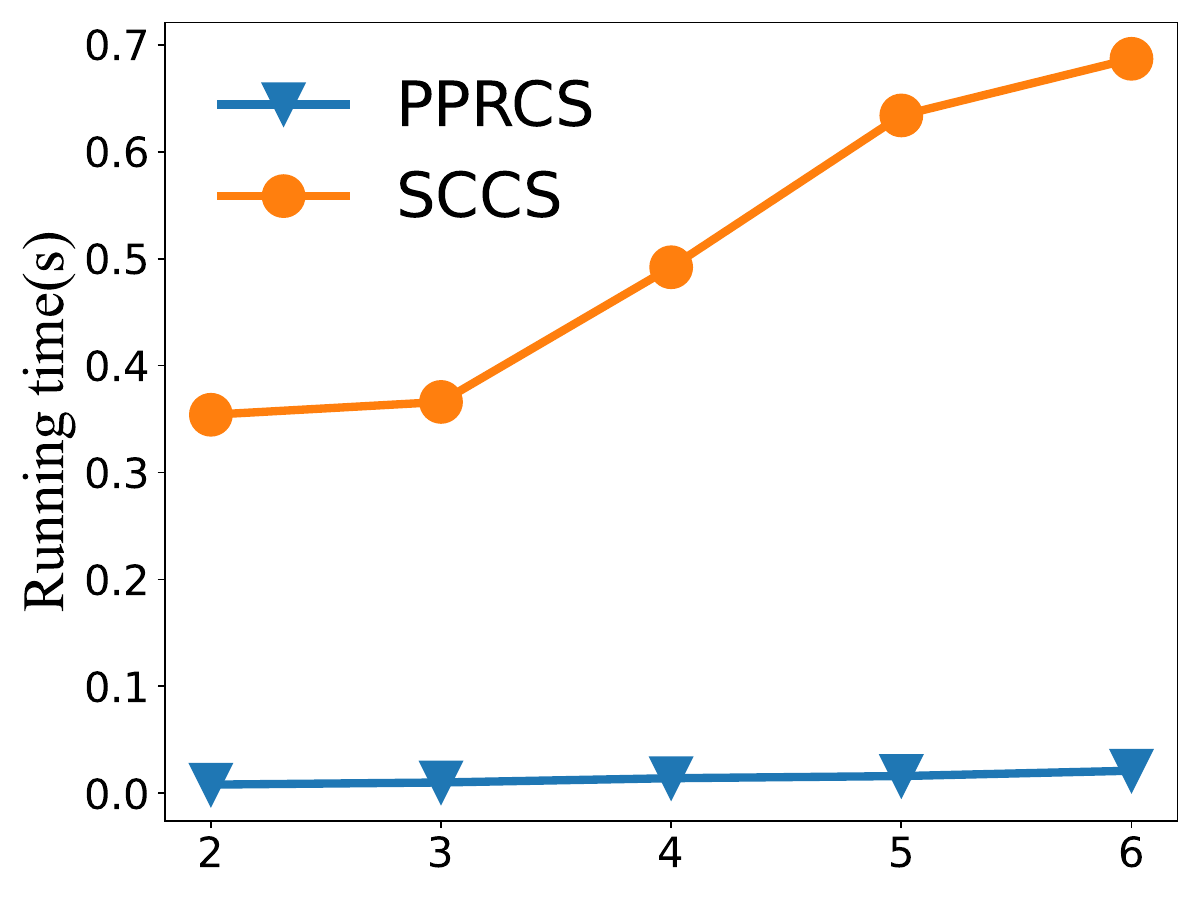}
\label{fig:LFR_n(b)}
}\vspace*{-0.3cm}
\caption{Performance on LFR with varying $n$ ($\times 10^4$).}
\label{fig:LFR_n} \vspace{-0.5cm}
\end{figure}

\stitle{Exp-8: Performance on LFRs with varying $ad$.} We fix $n = 30,000$ and $\mu = 0.2$, Figure \ref{fig:LFR_ad} reports the performance with $ad\in [8, 10 , 15, 20, 30]$. We have the following two aspects of observation: (1) The \textit{F1-score} improves with $ad$  increases and \textit{SCCS} consistently outperforms \textit{PPRCS}. It shows that on LFRs with higher $ad$, \textit{SCCS} is more effective in detecting community structures and performs more accurately in community partitioning. This phenomenon may be attributed to the denser connections between vertices on LFRs with higher average degrees, leading to stronger connections within communities, thereby facilitating the algorithm's ability to identify community boundaries and enhance the accuracy of community partitioning. (2) The running time of \textit{SCCS} and \textit{PPRCS} increases linearly with increasing $ad$. This may be due to denser connections between vertices in networks with higher average degrees, resulting in slower convergence of the algorithms as vertices consider more neighbors. Overall, \textit{SCCS} exhibits superior scalability compared to \textit{PPRCS} on LFRs with different average degrees.

\begin{figure}[t!]
\centering
\subfigure[\textit{F1-score}]{
\includegraphics[width=0.22\textwidth]{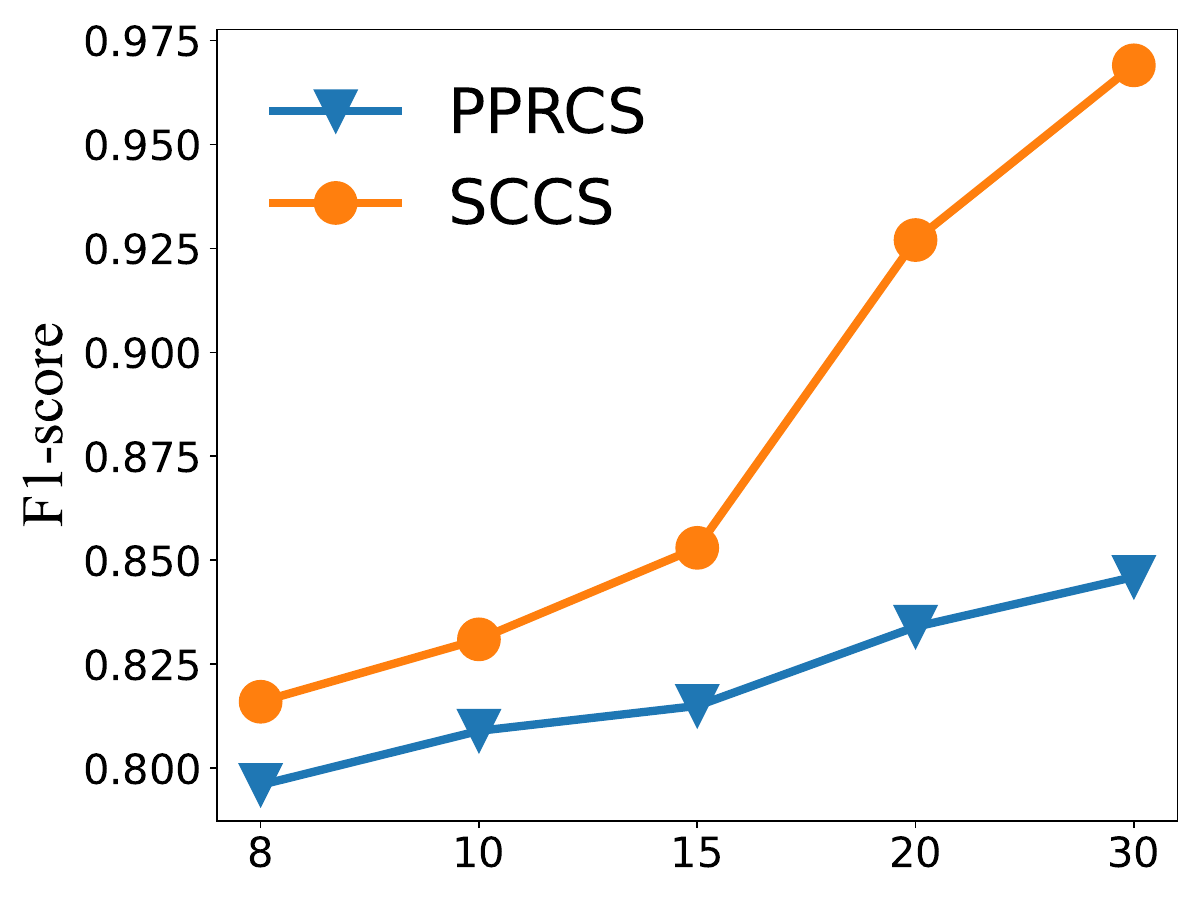}
\label{fig:LFR_ad(a)}}
\subfigure[\textit{Runtime}]{
\includegraphics[width=0.22\textwidth]{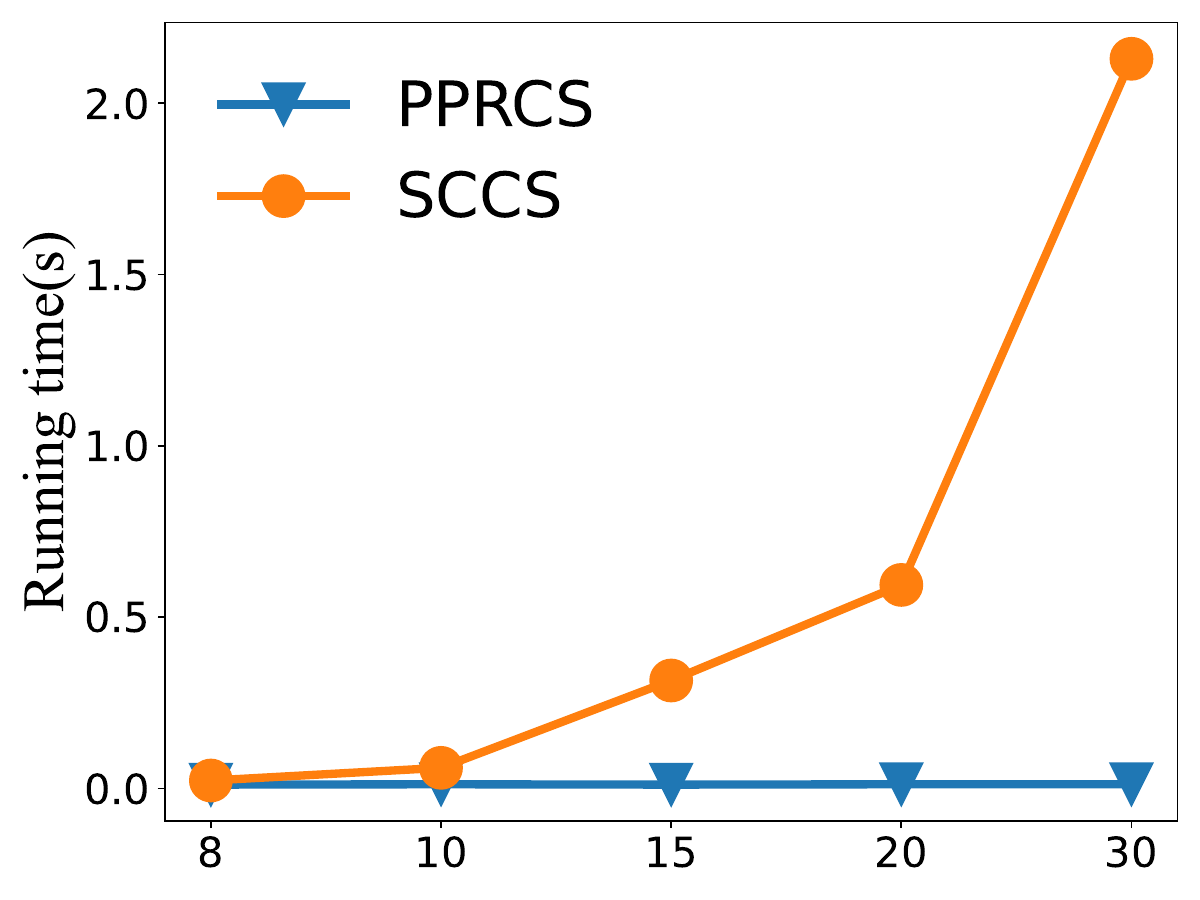}
\label{fig:LFR_ad(b)}
}\vspace*{-0.3cm}
\caption{Performance on LFR with varying  $ad$.}
\label{fig:LFR_ad} \vspace{-0.3cm}
\end{figure}

\stitle{Exp-9: Performance on LFRs with varying $\mu$.} As shown in Figure \ref{fig:LFR_mu}, we analyze LFR synthetic networks with  $n = 10,000$ and  $ad = 15$, varying $\mu$ from 0.1 with increments of 0.05 up to 0.3. The following conclusions can be drawn from Figure \ref{fig:LFR_mu}: (1) As $\mu$ increases, both the \textit{F1-score} of \textit{SCCS} and \textit{PPRCS} decrease, and their running times increase. This is because a higher $\mu$ results in a more complex network topology, making the distinction between communities more ambiguous. Consequently, the algorithms find it more challenging to accurately identify community structures, leading to increased computational complexity. (2) The running time of both \textit{SCCS} and \textit{PPRCS} remains short, with \textit{SCCS} consistently outperforming \textit{PPRCS} in terms of \textit{F1-score}. For instance, when $\mu=0.3$, both \textit{SCCS} and \textit{PPRCS} can obtain the target community within 0.5 seconds. At this point, the target community identified by \textit{SCCS} still maintains a high \textit{F1-score} of approximately 0.80, while the \textit{F1-score} for the community discovered by \textit{PPRCS} is around 0.58. Summing up, \textit{SCCS} remains effective and superior to \textit{PPRCS} in discovering community structures when $\mu \leq 0.3$.

\begin{figure}[t!]
\centering
\subfigure[\textit{F1-score}]{
\includegraphics[width=0.22\textwidth]{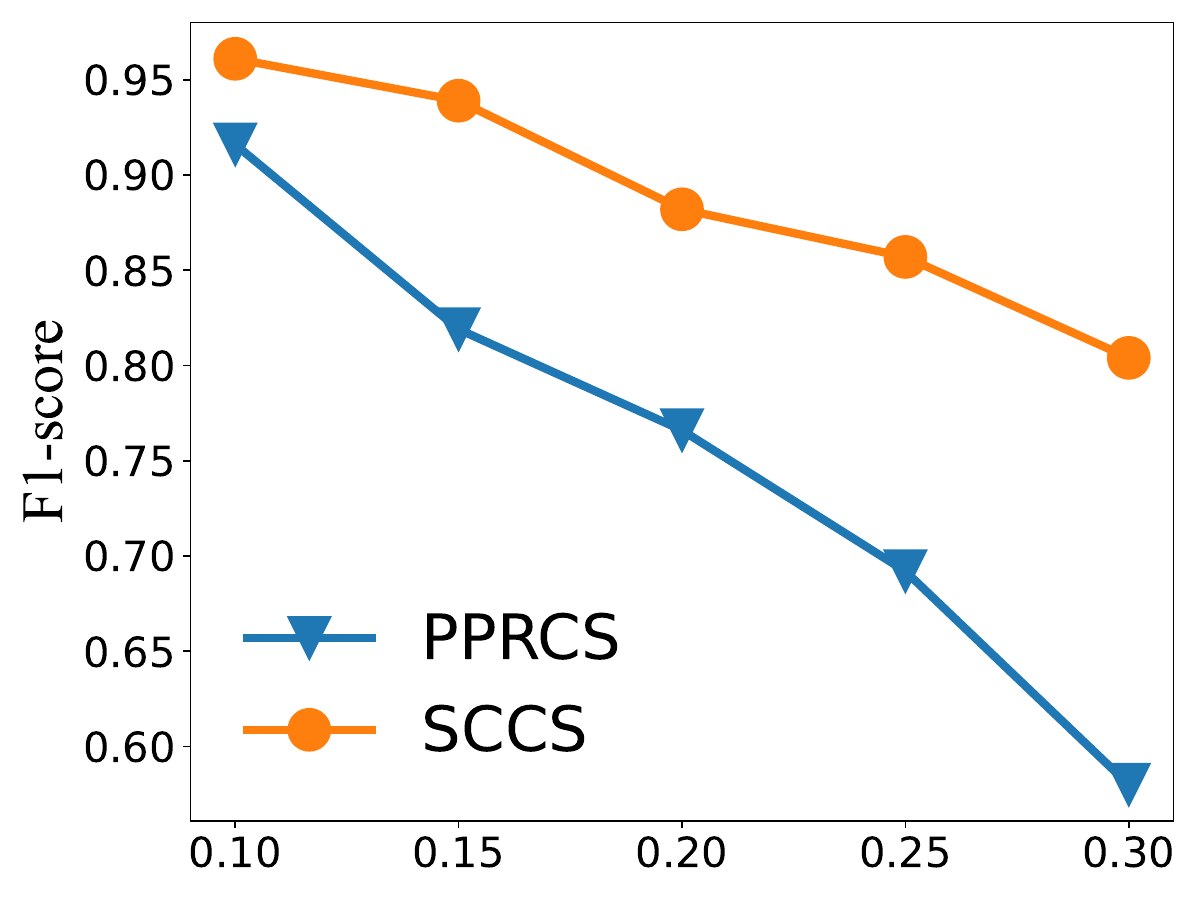}
\label{fig:LFR_mu(a)}
}
\subfigure[\textit{Runtime}]{
\includegraphics[width=0.22\textwidth]{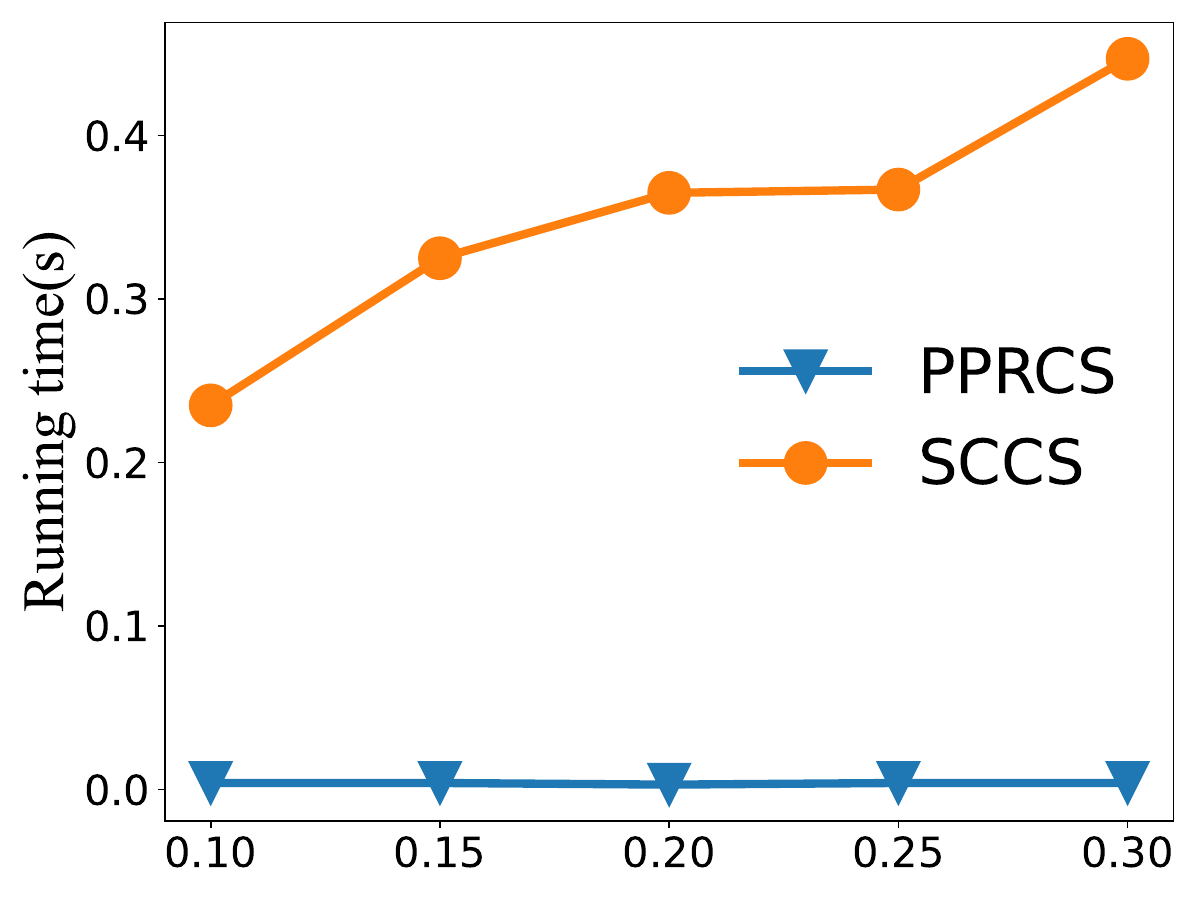}
\label{fig:LFR_mu(b)}
}\vspace*{-0.3cm}
\caption{Performance on LFR with varying $\mu$.}
\label{fig:LFR_mu} \vspace{-0.5cm}
\end{figure}

%% file: related_work.tex
\vspace{-0.3cm}
\section{Related Work}\label{sec:relate}

Community search has emerged as a prominent tool for analyzing networks \cite{DBLP:conf/icde/HuangLX17, DBLP:journals/vldb/FangHQZZCL20}, which primarily concentrates on locating the communities containing the user-specified query vertex. We can broadly classify community search into two categories: cohesive subgraph-based community search and conductance-based community search. 

\stitle{Cohesive subgraph-based community search}, which primarily emphasizes the internal connectivity of the community while overlooking the separation between communities \cite{DBLP:conf/kdd/SozioG10,DBLP:conf/sigmod/HuangCQTY14,DBLP:conf/sigmod/CuiXWLW13,DBLP:journals/pvldb/WuJLZ15}. For instance, Sozio et al. \cite{DBLP:conf/kdd/SozioG10} proposed a community search framework, wherein the desired community is defined as a connected subgraph containing query vertices and exhibiting a high score according to a specified quality metric, specifically utilizing the $k$-core as the quality metric. Given that the $k$-core might not always exhibit the desired density, Huang et al. \cite{DBLP:conf/sigmod/HuangCQTY14} opted for a more cohesive subgraph representation, utilizing the $k$-truss model to characterize the community. Recently, Wu et al. \cite{DBLP:journals/pvldb/WuJLZ15} noted a common issue with the aforementioned methods, known as the free-rider effect, where the returned community frequently includes numerous vertices irrelevant to the query vertices. To address this, they introduced the concept of query-biased density \cite{DBLP:journals/pvldb/WuJLZ15}, aiming to identify the densest subgraph in proximity to the query vertices. However, these methods overlook the external sparsity of the community, leading to sub-optimal results.

\stitle{Conductance-based community search}, which emphasizes both the internal connectivity of the community and the separation between communities \cite{DBLP:conf/focs/AndersenCL06, DBLP:conf/kdd/TongF06, DBLP:conf/icdm/BianYCWLZ18, DBLP:conf/sigmod/YangXWBZL19}. For example, Tong et al. \cite{DBLP:conf/kdd/TongF06} utilized random walk with restart to evaluate the goodness score of a vertex relative to the query vertices. Andersen et al. \cite{DBLP:conf/focs/AndersenCL06} employed Personalized PageRank to rank vertices and then applied a sweep-cut procedure to achieve the locally optimal \textit{conductance}.  Unfortunately, they suffer from several defects in terms of the community search problem (e.g., they cannot guarantee that the target community is connected \cite{DBLP:conf/icml/ZhuLM13} or contain the given query vertex \cite{DBLP:conf/webi/LuoWP06}). Recently, He et al. \cite{DBLP:journals/eswa/HeLYLJW24} combined self-defined vertex scores with perturbation strategies to avoid local optima, thereby obtaining communities with the minimal \textit{conductance}. Despite achieving promising results, it is only applicable to community searches with bounded size. Besides simple graphs, more informative attribute information assigned to vertices and edges has been explored for community search. For example, the community search on keyword-based graphs \cite{DBLP:journals/pvldb/FangCLH16, DBLP:journals/pvldb/HuangL17, DBLP:conf/icde/LiuZZHXG20}, multi-valued graphs \cite{DBLP:conf/sigmod/LiQYYXXZ18}, heterogeneous information networks \cite{DBLP:journals/pvldb/FangYZLC20, DBLP:journals/pvldb/JianWC20}, and temporal networks \cite{DBLP:journals/pvldb/LinYLZQJJ24, DBLP:conf/dasfaa/ZhangLYJ22}. Obviously, these methods  are orthogonal to our work.

%% file: sample-sigconf.bbl
\begin{thebibliography}{10}
\providecommand{\url}[1]{#1}
\csname url@samestyle\endcsname
\providecommand{\newblock}{\relax}
\providecommand{\bibinfo}[2]{#2}
\providecommand{\BIBentrySTDinterwordspacing}{\spaceskip=0pt\relax}
\providecommand{\BIBentryALTinterwordstretchfactor}{4}
\providecommand{\BIBentryALTinterwordspacing}{\spaceskip=\fontdimen2\font plus
\BIBentryALTinterwordstretchfactor\fontdimen3\font minus \fontdimen4\font\relax}
\providecommand{\BIBforeignlanguage}[2]{{%
\expandafter\ifx\csname l@#1\endcsname\relax
\typeout{** WARNING: IEEEtran.bst: No hyphenation pattern has been}%
\typeout{** loaded for the language `#1'. Using the pattern for}%
\typeout{** the default language instead.}%
\else
\language=\csname l@#1\endcsname
\fi
#2}}
\providecommand{\BIBdecl}{\relax}
\BIBdecl

\bibitem{DBLP:conf/focs/AndersenCL06}
R.~Andersen, F.~R.~K. Chung, and K.~J. Lang, ``Local graph partitioning using pagerank vectors,'' in \emph{FOCS}, 2006, pp. 475--486.

\bibitem{DBLP:conf/kdd/SozioG10}
M.~Sozio and A.~Gionis, ``The community-search problem and how to plan a successful cocktail party,'' in \emph{KDD}, 2010, pp. 939--948.

\bibitem{DBLP:conf/sigmod/CuiXWW14}
W.~Cui, Y.~Xiao, H.~Wang, and W.~Wang, ``Local search of communities in large graphs,'' in \emph{SIGMOD}, 2014, pp. 991--1002.

\bibitem{DBLP:conf/sigmod/YangXWBZL19}
R.~Yang, X.~Xiao, Z.~Wei, S.~S. Bhowmick, J.~Zhao, and R.~Li, ``Efficient estimation of heat kernel pagerank for local clustering,'' in \emph{SIGMOD}, 2019, pp. 1339--1356.

\bibitem{DBLP:conf/kdd/ChenL0XY020}
L.~Chen, C.~Liu, R.~Zhou, J.~Xu, J.~X. Yu, and J.~Li, ``Finding effective geo-social group for impromptu activities with diverse demands,'' in \emph{KDD}, 2020, pp. 698--708.

\bibitem{DBLP:journals/datamine/BarbieriBGG15}
N.~Barbieri, F.~Bonchi, E.~Galimberti, and F.~Gullo, ``Efficient and effective community search,'' \emph{Data Min. Knowl. Discov.}, vol.~29, no.~5, pp. 1406--1433, 2015.

\bibitem{DBLP:conf/sigmod/HuangCQTY14}
X.~Huang, H.~Cheng, L.~Qin, W.~Tian, and J.~X. Yu, ``Querying k-truss community in large and dynamic graphs,'' in \emph{SIGMOD}, 2014, pp. 1311--1322.

\bibitem{DBLP:conf/sigmod/LiuZ0XG20}
Q.~Liu, M.~Zhao, X.~Huang, J.~Xu, and Y.~Gao, ``Truss-based community search over large directed graphs,'' in \emph{SIGMOD}, 2020, pp. 2183--2197.

\bibitem{DBLP:journals/vldb/FangHQZZCL20}
Y.~Fang, X.~Huang, L.~Qin, Y.~Zhang, W.~Zhang, R.~Cheng, and X.~Lin, ``A survey of community search over big graphs,'' \emph{{VLDB} J.}, vol.~29, no.~1, pp. 353--392, 2020.

\bibitem{DBLP:conf/www/LeskovecLM10}
J.~Leskovec, K.~J. Lang, and M.~W. Mahoney, ``Empirical comparison of algorithms for network community detection,'' in \emph{WWW}, M.~Rappa, P.~Jones, J.~Freire, and S.~Chakrabarti, Eds., 2010, pp. 631--640.

\bibitem{DBLP:journals/kais/YangL15}
J.~Yang and J.~Leskovec, ``Defining and evaluating network communities based on ground-truth,'' \emph{Knowl. Inf. Syst.}, vol.~42, no.~1, pp. 181--213, 2015.

\bibitem{DBLP:journals/kais/ShinEF18}
K.~Shin, T.~Eliassi{-}Rad, and C.~Faloutsos, ``Patterns and anomalies in k-cores of real-world graphs with applications,'' \emph{Knowl. Inf. Syst.}, vol.~54, no.~3, pp. 677--710, 2018.

\bibitem{DBLP:conf/www/LeskovecLDM08}
J.~Leskovec, K.~J. Lang, A.~Dasgupta, and M.~W. Mahoney, ``Statistical properties of community structure in large social and information networks,'' in \emph{WWW}.\hskip 1em plus 0.5em minus 0.4em\relax {ACM}, 2008, pp. 695--704.

\bibitem{DBLP:journals/tbd/WangYBH24}
M.~Wang, Y.~Yang, D.~Bindel, and K.~He, ``Streaming local community detection through approximate conductance,'' \emph{{IEEE} Trans. Big Data}, vol.~10, no.~1, pp. 12--22, 2024.

\bibitem{DBLP:conf/icdm/BianYCWLZ18}
Y.~Bian, Y.~Yan, W.~Cheng, W.~Wang, D.~Luo, and X.~Zhang, ``On multi-query local community detection,'' in \emph{ICDM}, 2018, pp. 9--18.

\bibitem{DBLP:conf/aaai/LinLJ23}
L.~Lin, R.~Li, and T.~Jia, ``Scalable and effective conductance-based graph clustering,'' in \emph{AAAI}, 2023.

\bibitem{DBLP:conf/icml/ZhuLM13}
Z.~A. Zhu, S.~Lattanzi, and V.~S. Mirrokni, ``A local algorithm for finding well-connected clusters,'' in \emph{ICML}, 2013, pp. 396--404.

\bibitem{DBLP:conf/webi/LuoWP06}
F.~Luo, J.~Z. Wang, and E.~Promislow, ``Exploring local community structures in large networks,'' in \emph{WI}.\hskip 1em plus 0.5em minus 0.4em\relax {IEEE} Computer Society, 2006, pp. 233--239.

\bibitem{DBLP:conf/kdd/KlosterG14}
K.~Kloster and D.~F. Gleich, ``Heat kernel based community detection,'' in \emph{{KDD}}, 2014, pp. 1386--1395.

\bibitem{DBLP:conf/kdd/Lin0WZL24}
L.~Lin, T.~Jia, Z.~Wang, J.~Zhao, and R.~Li, ``{PSMC:} provable and scalable algorithms for motif conductance based graph clustering,'' in \emph{KDD}, 2024, pp. 1793--1803.

\bibitem{DBLP:journals/eswa/HeLYLJW24}
Y.~He, L.~Lin, P.~Yuan, R.~Li, T.~Jia, and Z.~Wang, ``{CCSS:} towards conductance-based community search with size constraints,'' \emph{Expert Syst. Appl.}, vol. 250, p. 123915, 2024.

\bibitem{DBLP:books/fm/GareyJ79}
M.~R. Garey and D.~S. Johnson, \emph{Computers and Intractability: {A} Guide to the Theory of NP-Completeness}.\hskip 1em plus 0.5em minus 0.4em\relax W. H. Freeman, 1979.

\bibitem{csm1}
C.~Li, F.~Zhang, Y.~Zhang, L.~Qin, W.~Zhang, and X.~Lin, ``Efficient progressive minimum k-core search,'' \emph{Proc. {VLDB} Endow.}, vol.~13, no.~3, pp. 362--375, 2019.

\bibitem{truss1}
E.~Akbas and P.~Zhao, ``Truss-based community search: a truss-equivalence based indexing approach,'' \emph{Proc. {VLDB} Endow.}, vol.~10, no.~11, pp. 1298--1309, 2017.

\bibitem{truss2}
Y.~Jiang, X.~Huang, and H.~Cheng, ``{I/O} efficient k-truss community search in massive graphs,'' \emph{{VLDB} J.}, vol.~30, no.~5, pp. 713--738, 2021.

\bibitem{ppr1}
H.~Avron and L.~Horesh, ``Community detection using time-dependent personalized pagerank,'' in \emph{ICML}, vol.~37, 2015, pp. 1795--1803.

\bibitem{hk1}
Z.~Lu, J.~Wahlstr{\"o}m, and A.~Nehorai, ``Local clustering via approximate heat kernel pagerank with subgraph sampling,'' \emph{Scientific reports}, vol.~11, no.~1, p. 15786, 2021.

\bibitem{hk2}
R.~Yang, X.~Xiao, Z.~Wei, S.~S. Bhowmick, J.~Zhao, and R.~Li, ``Efficient estimation of heat kernel pagerank for local clustering,'' in \emph{SIGMOD}, 2019, pp. 1339--1356.

\bibitem{trevisan2017lecture}
L.~Trevisan, ``Lecture notes on graph partitioning, expanders and spectral methods,'' \emph{University of California, Berkeley, https://people. eecs. berkeley. edu/luca/books/expanders-2016. pdf}, 2017.

\bibitem{DBLP:journals/kbs/DingZY18}
X.~Ding, J.~Zhang, and J.~Yang, ``A robust two-stage algorithm for local community detection,'' \emph{Knowl. Based Syst.}, vol. 152, pp. 188--199, 2018.

\bibitem{DBLP:journals/tkdd/HeSBH19}
K.~He, P.~Shi, D.~Bindel, and J.~E. Hopcroft, ``Krylov subspace approximation for local community detection in large networks,'' \emph{{ACM} Trans. Knowl. Discov. Data}, vol.~13, no.~5, pp. 52:1--52:30, 2019.

\bibitem{macqueen2001community}
K.~M. MacQueen, E.~McLellan, D.~S. Metzger, S.~Kegeles, R.~P. Strauss, R.~Scotti, L.~Blanchard, and R.~T. Trotter, ``What is community? an evidence-based definition for participatory public health,'' \emph{American journal of public health}, vol.~91, no.~12, pp. 1929--1938, 2001.

\bibitem{DBLP:conf/icdm/YangL12}
J.~Yang and J.~Leskovec, ``Defining and evaluating network communities based on ground-truth,'' in \emph{ICDM}, 2012, pp. 745--754.

\bibitem{nx}
A.~A. Hagberg, D.~A. Schult, and P.~J. Swart, ``Exploring network structure, dynamics, and function using networkx,'' \emph{In Proceedings of the 7th Python in Science Conference}, vol.~2, no.~1, pp. 11--15, 2008.

\bibitem{lancichinetti2009detecting}
A.~Lancichinetti, S.~Fortunato, and J.~Kert{\'e}sz, ``Detecting the overlapping and hierarchical community structure in complex networks,'' \emph{New journal of physics}, vol.~11, no.~3, p. 033015, 2009.

\bibitem{clauset2005finding}
A.~Clauset, ``Finding local community structure in networks,'' \emph{Physical review E}, vol.~72, no.~2, p. 026132, 2005.

\bibitem{DBLP:journals/pvldb/WuJLZ15}
Y.~Wu, R.~Jin, J.~Li, and X.~Zhang, ``Robust local community detection: On free rider effect and its elimination,'' \emph{Proc. {VLDB} Endow.}, vol.~8, no.~7, pp. 798--809, 2015.

\bibitem{DBLP:conf/icde/HuangLX17}
X.~Huang, L.~V.~S. Lakshmanan, and J.~Xu, ``Community search over big graphs: Models, algorithms, and opportunities,'' in \emph{ICDE}, 2017.

\bibitem{DBLP:conf/sigmod/CuiXWLW13}
W.~Cui, Y.~Xiao, H.~Wang, Y.~Lu, and W.~Wang, ``Online search of overlapping communities,'' in \emph{SIGMOD}, 2013, pp. 277--288.

\bibitem{DBLP:conf/kdd/TongF06}
H.~Tong and C.~Faloutsos, ``Center-piece subgraphs: problem definition and fast solutions,'' in \emph{KDD}, 2006, pp. 404--413.

\bibitem{DBLP:journals/pvldb/FangCLH16}
Y.~Fang, R.~Cheng, S.~Luo, and J.~Hu, ``Effective community search for large attributed graphs,'' \emph{Proc. {VLDB} Endow.}, vol.~9, no.~12, pp. 1233--1244, 2016.

\bibitem{DBLP:journals/pvldb/HuangL17}
X.~Huang and L.~V.~S. Lakshmanan, ``Attribute-driven community search,'' \emph{Proc. {VLDB} Endow.}, vol.~10, no.~9, pp. 949--960, 2017.

\bibitem{DBLP:conf/icde/LiuZZHXG20}
Q.~Liu, Y.~Zhu, M.~Zhao, X.~Huang, J.~Xu, and Y.~Gao, ``{VAC:} vertex-centric attributed community search,'' in \emph{ICDE}, 2020, pp. 937--948.

\bibitem{DBLP:conf/sigmod/LiQYYXXZ18}
R.~Li, L.~Qin, F.~Ye, J.~X. Yu, X.~Xiao, N.~Xiao, and Z.~Zheng, ``Skyline community search in multi-valued networks,'' in \emph{SIGMOD}, 2018.

\bibitem{DBLP:journals/pvldb/FangYZLC20}
Y.~Fang, Y.~Yang, W.~Zhang, X.~Lin, and X.~Cao, ``Effective and efficient community search over large heterogeneous information networks,'' \emph{Proc. {VLDB} Endow.}, vol.~13, no.~6, pp. 854--867, 2020.

\bibitem{DBLP:journals/pvldb/JianWC20}
X.~Jian, Y.~Wang, and L.~Chen, ``Effective and efficient relational community detection and search in large dynamic heterogeneous information networks,'' \emph{Proc. {VLDB} Endow.}, vol.~13, no.~10, pp. 1723--1736, 2020.

\bibitem{DBLP:journals/pvldb/LinYLZQJJ24}
L.~Lin, P.~Yuan, R.~Li, C.~Zhu, H.~Qin, H.~Jin, and T.~Jia, ``{QTCS:} efficient query-centered temporal community search,'' \emph{Proc. {VLDB} Endow.}, vol.~17, no.~6, pp. 1187--1199, 2024.

\bibitem{DBLP:conf/dasfaa/ZhangLYJ22}
Y.~Zhang, L.~Lin, P.~Yuan, and H.~Jin, ``Significant engagement community search on temporal networks,'' in \emph{DASFAA}, vol. 13245, 2022, pp. 250--258.

\end{thebibliography}
